%% file: main.tex
\documentclass[11pt]{article}

\usepackage{amsthm}
\usepackage{graphicx} % support the \includegraphics command and options
\usepackage{array} % for better arrays (eg matrices) in maths

\usepackage{amsmath, amssymb, amsfonts, verbatim}
\usepackage{hyphenat,epsfig,subcaption,multirow}

\usepackage[usenames,dvipsnames]{xcolor}
\usepackage[ruled]{algorithm2e}

\usepackage{tcolorbox}
\tcbuselibrary{skins,breakable}
\tcbset{enhanced jigsaw}

\usepackage[compact]{titlesec}

\definecolor{DarkRed}{rgb}{0.5,0.1,0.1}
\definecolor{DarkBlue}{rgb}{0.1,0.1,0.5}

\usepackage{nameref}
\definecolor{ForestGreen}{rgb}{0.1333,0.5451,0.1333}
\definecolor{Red}{rgb}{0.9,0,0}
\usepackage[linktocpage=true,
	pagebackref=true,colorlinks,
	linkcolor=DarkRed,citecolor=ForestGreen,
	bookmarks,bookmarksopen,bookmarksnumbered]
	{hyperref}

\usepackage{bm}
\usepackage{url}
\usepackage{xspace}
\usepackage[mathscr]{euscript}

\usepackage{tikz}
\usetikzlibrary{arrows}
\usetikzlibrary{arrows.meta}
\usetikzlibrary{shapes}
\usetikzlibrary{backgrounds}
\usetikzlibrary{positioning}
\usetikzlibrary{decorations.markings}
\usetikzlibrary{patterns}
\usetikzlibrary{calc}
\usetikzlibrary{fit}
\usetikzlibrary{snakes}

\usepackage{mdframed}

\usepackage[noend]{algpseudocode}
\makeatletter
\def\BState{\State\hskip-\ALG@thistlm}
\makeatother

\usepackage{cite}
\usepackage{enumitem}

\usepackage[margin=0.8in]{geometry}

\newtheorem{theorem}{Theorem}
\newtheorem{lemma}{Lemma}[section]

\newtheorem{claim}[lemma]{Claim}

\newtheorem{definition}{Definition}
\newtheorem*{Definition}{Definition}
\newtheorem{problem}{Problem}
\newtheorem{remark}[lemma]{Remark}

\newtheorem*{claim*}{Claim}
\newtheorem*{proposition*}{Proposition}
\newtheorem*{lemma*}{Lemma}
\newtheorem*{problem*}{Problem}

\newtheorem{mdresult}{Result}

\newtheorem{mdinvariant}{Invariant}

\newcommand{\ssection}[1]{\paragraph{#1}}

\allowdisplaybreaks

\renewcommand{\qed}{\nobreak \ifvmode \relax \else
      \ifdim\lastskip<1.5em \hskip-\lastskip
      \hskip1.5em plus0em minus0.5em \fi \nobreak
      \vrule height0.75em width0.5em depth0.25em\fi}

\newcommand*\samethanks[1][\value{footnote}]{\footnotemark[#1]}

\input{macros}

\title{Distributed Weighted Matching via Randomized Composable Coresets\footnote{An extended abstract appears in \emph{Proceedings of the 36th International Conference on Machine Learning (ICML) 2019}.}}
\author{Sepehr Assadi\thanks{Department of Computer Science, Princeton University; supported in part by the Simons Collaboration on Algorithms and Geometry. 
Majority of the work done while the author was a summer intern at Google Research, New York. Emails: {{\small {\tt sassadi@princeton.edu.}}}}
\and MohammadHossein Bateni\thanks{Google Research, New York, NY, US. Email: {{\small {\tt \{bateni,mirrokni\}@google.com.}}}}
\and Vahab Mirrokni\samethanks
}

\date{}

\begin{document}
\maketitle

%%\pagenumbering{roman}
%%
%%\input{abstract}
%%
%%\clearpage
%%
%%\setcounter{tocdepth}{3}
%%\tableofcontents
%%
%%%%
%%\clearpage
%%
%%\pagenumbering{arabic}
%%\setcounter{page}{1}

\begin{abstract}
  Maximum weight matching is one of the most fundamental combinatorial optimization problems with a wide range of applications in data mining and bioinformatics.  Developing distributed weighted matching algorithms is challenging due to the sequential
  nature of efficient algorithms for this problem.  In this paper, we develop a simple distributed algorithm for the problem on general graphs with approximation guarantee of $2+\epsilon$ that (nearly) {\em matches} that of the sequential {\em greedy} algorithm. A key
  advantage of this algorithm is that it can be easily implemented in only {\em two rounds} of computation in modern parallel computation frameworks such as MapReduce. We also demonstrate the efficiency of our algorithm in practice on various graphs
  (some with half a trillion edges) by achieving objective values always close to what is achievable in the centralized setting.  
\end{abstract}

\input{intro}

\input{prelim}

\input{coreset}

%%
\input{weighted-coreset}

%%%%
\input{app-exp}

\bibliographystyle{abbrv}
\bibliography{general}

\clearpage 

\appendix
\input{app-optimal}

\end{document}

%% file: macros.tex
\newcommand{\toShrink}{-.20cm}
\newcommand{\toShrinkEnu}{-.2cm}

\newenvironment{tbox}{\begin{tcolorbox}[
		enlarge top by=5pt,
		enlarge bottom by=5pt,
		 boxsep=0pt,
                  left=4pt,
                  right=4pt,
                  top=10pt,
                  arc=0pt,
                  boxrule=1pt,toprule=1pt,
                  colback=white
                  ]%%
	}
{\end{tcolorbox}}

\newcommand{\Geq}[1]{\ensuremath{\underset{\textnormal{#1}}\geq}}

\newcommand{\Ot}{\ensuremath{\widetilde{O}}}
\newcommand{\eps}{\ensuremath{\varepsilon}}
\newcommand{\Paren}[1]{\Big(#1\Big)}
\newcommand{\Bracket}[1]{\Big[#1\Big]}
\newcommand{\bracket}[1]{\left[#1\right]}
\newcommand{\paren}[1]{\ensuremath{\left(#1\right)}\xspace}
\newcommand{\card}[1]{\left\vert{#1}\right\vert}

\newcommand{\set}[1]{\ensuremath{\left\{ #1 \right\}}}

\newcommand{\polylog}{\mbox{\rm  polylog}}

\newcommand{\opt}{\textnormal{\ensuremath{\mbox{opt}}}\xspace}

\newcommand{\alg}{\ensuremath{\mathcal{A}}\xspace}

\DeclareMathOperator*{\Exp}{\ensuremath{{\mathbb{E}}}}
\DeclareMathOperator*{\Prob}{\ensuremath{\textnormal{Pr}}}
\renewcommand{\Pr}{\Prob}

\newcommand{\Ex}{\Exp}

\newcommand{\etal}{et al.\xspace}

\newcommand{\Ei}[1]{\ensuremath{E^{(#1)}}}

\newcommand{\Gi}[1]{\ensuremath{G^{(#1)}}}

\renewcommand{\alg}{\ensuremath{\textnormal{\textsf{ALG}}}\xspace}

\newcommand{\Greedy}{\ensuremath{\textnormal{\textsf{Greedy}}}\xspace}
\newcommand{\GreedyCoreset}{\ensuremath{\textnormal{\textsf{GreedyCoreset}}}\xspace}

\newcommand{\Mstar}{\ensuremath{M^{*}}}

\newcommand{\ww}[1]{\ensuremath{w(#1)}}

%% file: intro.tex
\section{Introduction}\label{sec:intro}

A matching in a graph is defined as a collection of edges that do not share any vertices. The problem of finding a matching with a maximum weight in an edge-weighted graph---henceforth referred to
as the \emph{maximum weight matching} (MWM) problem---is one of the most fundamental combinatorial optimization problems with a wide range of applications in data mining and bioinformatics.
For instance, maximum weight matchings can improve the quality of data clustering~\cite{BBDHKLM17} or partitioning~\cite{KK98}, as well as discovery of subgraphs in networks in bioinformatics~\cite{LD04,BSX08}. 
Other  applications are in trading markets and computational advertising~\cite{PennT00,MSVV07,CharlesCDJS10}, kidney exchange~\cite{DickersonPS12,BlumDHPSS15}, online labor markets~\cite{BehnezhadR18}, 
and semi-supervised learning~\cite{JebaraWC09} (see~\cite{ManshadiAGKMS13} for other similar examples). 
Yet another application arise in numerical linear algebra, e.g., in sparse linear solvers~\cite{DK99,DK01}, decomposition of sparse matrices~\cite{PF90}, and computing sparse bases for underdetermined matrices~\cite{PCP06}.

The study of MWM
dates back to the introduction of the complexity class $\mathbf{P}$ as the set of ``tractable'' problems by Edmonds~\cite{Edmonds65a,Edmonds65} who designed a poly-time algorithm for this problem
on general graphs. Since then, there have been numerous attempts in developing faster algorithms for MWM (see, e.g.,~\cite{Gabow76,GabowGS84,Gabow85,Gabow90,GabowT91,CyganGS12,DuanPS17}) 
culminating in the $\Ot(m\sqrt{n})$-time\footnote{Throughout the paper, we use $\Ot(f):= O(f) \cdot \polylog{(f)}$.}
algorithm of Gabow and Tarjan~\cite{GabowT91}. For approximation algorithms, the greedy algorithm that repeatedly picks the heaviest edge possible in the matching achieves a 
two approximation and after a series of work~\cite{Preis99,VinkemeierH05,PettieS04,DuanP10}, an $\Ot(m/\eps)$-time algorithm for $(1+\eps)$-approximation was developed by Duan and Pettie~\cite{DuanP14}. 

Nowadays, many applications that involve MWM require processing massive graphs that are typically being stored and processed in a distributed fashion. Classical algorithmic approaches to MWM are no longer viable options to cope 
with challenges that stem from processing massive graphs and one should now instead focus on algorithms that can be implemented \emph{efficiently in distributed
settings} even at the cost of a (slight) reduction in the quality of the solution. 

In this paper, we design a simple and efficient greedy distributed algorithm for the maximum weight matching problem with an approximation ratio that nearly matches that of the {sequential greedy algorithm} for MWM. 
Our algorithm can also be easily implemented in two rounds of parallel computation in MapReduce-style computation frameworks, which is known to be the minimum number of rounds necessary for solving this problem. 

%%Maximum weight matching admits a simple $2$-approximation \emph{greedy} algorithm: traverse the edges from the heaviest to lightest and add each edge greedily to the matching
%%whenever possible, i.e., when both endpoints of the edge are currently unmatched. It is a simple exercise to verify that this algorithm indeed outputs a $2$-approximation
%%on any graph. 
%%
%%
%% 

%%For an input graph $G(V,E)$ with $n:= \card{V}$ and $m:= \card{E}$, the input size is $\Theta(m)$. In order for space requirements of our algorithms to fit the restrictions of the MPC model, we assume that $m = n^{1+\Omega(1)}$. This is the common
%%assumption in all previous work that achieve a round complexity of better than $O(\log{n})$ for any graph problem~\cite{KarloffSV10,CzumajLMMOS17,LattanziMSV11,AhnG15,AssadiK17}. 

\subsection{Background and Related Work}\label{sec:background}

Maximum weight and maximum cardinality matchings have been studied extensively in different models of computation for processing massive graphs such as streaming and distributed settings; see, 
e.g.,~\cite{McGregor05,EpsteinLMS11,GoelKK12,KonradMM12,AhnGM12,AhnG13,Kapralov13,ManshadiAGKMS13,CrouchS14,AssadiKLY16,AssadiKL17,PazS17,AhnG15,LattanziMSV11,HuangRVZ15,AssadiK17,CzumajLMMOS18,HarveyLL18,AssadiBBMS17} and
 references therein. 

\begin{table*}[t!]
\caption{A summary of previous work on MapReduce algorithms for MWM and our result.  \label{tab:comp}}

\centering
{\def\arraystretch{1.5}\tabcolsep=10pt
\begin{tabular}{cccc}
\hline
{\bf Reference} & {\bf Approximation} & {\bf Memory Per-Machine} & {\bf Rounds} \\
\hline
\cite{LattanziMSV11} & $8$ & $\Ot(n)$ & $O(\log{n})$ \\
\cite{CrouchS14} & $4$ & $\Ot(n)$ & $O(\log{n})$ \\
\cite{AhnG15} & $1+\eps$ & $\Ot(n)$ & $O(\eps^{-1}\log{n})$ \\
\cite{HarveyLL18} & $2$ & $\Ot(n)$ & $O(\log{n})$ \\
\cite{CzumajLMMOS18} & $2+\eps$ & $\Ot(n)$ & $O(\eps^{-\Theta(1/\eps)} \cdot \paren{\log\log{n}}^2)$ \\
\cite{AssadiBBMS17} & $2+\eps$ & $\Ot(n)$ & $O(\eps^{-\Theta(1/\eps)} \cdot \log\log{n})$ \\
\cite{GamlathKMS18} & $1+\eps$ & $\Ot(n)$ & $O(\eps^{-\Theta(1/\eps^2)} \cdot \log\log{n})$\\
\hline
\cite{LattanziMSV11} & $8$ & $n^{1+\Omega(1)}$ & $O(1)$ \\
\cite{CrouchS14} & $4$ & $n^{1+\Omega(1)}$ & $O(1)$ \\
\cite{AhnG15} & $1+\eps$ & $n^{1+\Omega(1)}$ & $O(1/\eps)$ \\
\cite{HarveyLL18} & $2$ & $n^{1+\Omega(1)}$ & $O(1)$ \\
\cite{AssadiK17} & $O(1)$ & $\Ot(n\sqrt{n})$ & $2$  \\
\cite{AssadiBBMS17} & $3+\eps$ & $\Ot(n\sqrt{n})$ & $2$  \\
\hline
\textbf{This paper} & $2+\eps$ & $O(n\sqrt{n})$ & $2$  \\
\end{tabular}
}
\end{table*}

Closely related to our work, Lattanzi~\etal~\cite{LattanziMSV11} designed MapReduce algorithms with $2$- and $8$-approximation guarantees, respectively, for unweighted and weighted matchings in $O(1)$ rounds 
on machines with memory $n^{1+\Omega(1)}$. These results were subsequently improved to $(1+\eps)$-approximation for both problems in $O(1/\eps)$ rounds by Ahn and Guha~\cite{AhnG15} using sophisticated primal-dual algorithms 
and multiplicative-weight-update method (for \emph{unweighted bipartite} matching, 
a simpler algorithm with $(1+\eps)$-approximation in $O(1/\eps)$ rounds using $O(n\sqrt{n})$ space was recently proposed in~\cite{BehnezhadDETY17}). 
Very recently, Harvey~\etal~\cite{HarveyLL18} 
designed a $2$-approximation algorithm for weighted matchings in $O(1)$ rounds based on the local-ratio theorem of~\cite{PazS17} for MWM.
Furthermore, Assadi and Khanna~\cite{AssadiK17} designed
a MapReduce algorithm with $O(n\sqrt{n})$ memory---using the so-called randomized composable coreset method which we also exploit in this paper---that achieves an $O(1)$-approximation to both problems in only two rounds of computation which is the 
optimal number of rounds by a result of~\cite{AssadiKLY16}. This result was very recently improved by Assadi~\etal~\cite{AssadiBBMS17} to (almost) $1.5$-approximation for unweighted matchings. 
Recent papers by~\cite{CzumajLMMOS18,AssadiBBMS17,GhaffariGKMR18} also considered these problems with
smaller per-machine memory and achieved $(1+\eps)$- and $(2+\eps)$-approximation for unweighted and weighted matchings in $O(\log\log{n})$ rounds and $\Ot(n)$ memory per-machine. 
The approximation ratio for weighted matchings in these results were very recently improved to by $(1+\eps)$~\cite{GamlathKMS18}.

Our work is also closely aligned with the trend on ``parallelizing'' sequential greedy algorithms in distributed settings (e.g.,~\cite{KumarMVV13,MirrokniZ15,BarbosaENW15,BarbosaENW16,HarveyLL18}). 
As noted elegantly by Kumar~\etal~\cite{KumarMVV13}: ``Greedy 
algorithms are practitioners' best friends---they are intuitive, simple to implement, and often lead to very good solutions. However, implementing greedy algorithms in a distributed setting is challenging since the greedy choice is inherently sequential, and it is not 
clear how to take advantage of the extra processing power.''  As such there have been extensive efforts in recent years to carry over the greedy algorithms in the sequential setting to distributed models as well.
These results are typically of two types: they either use a relatively large number of rounds to ``faithfully'' simulate the greedy algorithm, i.e., to obtain approximation guarantees that (almost) match that of the greedy 
algorithm~\cite{KumarMVV13,BarbosaENW16}, or use a very small number of rounds, say one or two,  for ``weak'' simulation, resulting in approximation guarantees that are within some constant factor of the corresponding greedy 
algorithm~\cite{MirrokniZ15,BarbosaENW15}. Table~\ref{tab:comp} provides a succinct summary of previous work.

\subsection{Our Contribution}\label{sec:results}

In this paper, we ``faithfully'' parallelize the sequential greedy algorithm for MWM in two rounds of parallel computation. In particular, we present an algorithm in 
the MapReduce framework (defined formally in Section~\ref{sec:prelim}) that for any constant $\eps > 0$, outputs a $(2+\eps)$-approximation to maximum weight matching in expectation using $O(\sqrt{\frac{m}{n}})$ 
machines	each with $O(\sqrt{mn})$ memory and in only two rounds of computation; here, $m$ and $n$ denote the number of edges and vertices in the graph, respectively. See Theorem~\ref{thm:main} for the formal statement of this result. 

Our distributed algorithm works as follows: send each edge of the graph to $O(1)$ machines \emph{randomly}, run the greedy algorithm---the one that repeatedly picks the heaviest available edge in the matching---on each part separately, combine the output
of the greedy algorithms on a single machine, and find a near-optimal weighted matching among these edges using any standard offline algorithm, say algorithm of~\cite{DuanP14} (see also
Section~\ref{sec:wm-mapreduce} for details on when one can simply run the greedy algorithm at the end). We prove that this simple algorithm leads to an almost two approximate matching of the original graph. 
This technique of partitioning the input randomly and computing a subgraph of each piece (here, a matching output by the greedy algorithm) 
is called the \emph{randomized composable coreset} technique and has been used previously in context of unweighted matchings~\cite{AssadiK17,AssadiBBMS17} and constrained submodular
maximization~\cite{MirrokniZ15,BarbosaENW15} (see Section~\ref{sec:coreset}). 
Finally, the number of rounds of our algorithm is \emph{optimal} by a result of~\cite{AssadiKLY16}. 
%% imply that getting even a very 
%%weak approximation ratio of $n^{o(1)}$ to maximum cardinality matching (when all weights are equal) is not possible in one round of MapReduce computation. 

\ssection{Comparison with prior work.} We conclude this section by making the following two comparisons: % between our result and previous results in the literature: 
	
\vspace{-5pt}
\begin{itemize}[leftmargin=10pt]
	 
	\item Number of rounds used by our algorithm is an \textbf{absolute constant two}, independent of the approximation of the algorithm. This significantly improves upon the previously best
	 MapReduce algorithms of~\cite{AhnG15,HarveyLL18} that require a 
	\textbf{large unspecified constant} number of rounds to achieve a similar guarantee on the approximation ratio. Other 
	algorithms for weighted matching 
	with similar guarantee on the number of rounds as ours are that of~\cite{LattanziMSV11} that achieves $8$-approximation (improvable to (almost) $4$-approximation using the Crouch-Stubbs technique~\cite{CrouchS14}) in \emph{six} rounds
	when using the same per-machine memory as ours (and \emph{at least three} rounds by allowing even more memory), and (almost) $3$-approximation of~\cite{AssadiBBMS17} which is based on a considerably complicated algorithm
	(as it first finds an approximation to \emph{unweighted} matching with better than $2$-approximation which is well-known to be a ``hard task'' for matchings\footnote{For instance, getting efficient algorithms with better than $2$-approximation in both streaming and dynamic graphs models are longstanding open problems, and very recently proven to be impossible for online edge-arrival graphs~\cite{GamlathKMSW19}.}). 
	We emphasize that the main bottleneck in MapReduce computation is the transition between different rounds (see, e.g.,~\cite{LattanziMSV11}) and hence 
	minimizing the number of rounds is the primary goal in this setting. 

	\item Previous work on parallelizing greedy algorithms in MapReduce framework suffered from one of the following two drawbacks: either a \textbf{suboptimal approximation guarantee} compared to the greedy
	algorithm {even by allowing unbounded computation time on each machine}~\cite{BarbosaENW15,MirrokniZ15}, or a \textbf{relatively large number of rounds} to match the performance of the greedy algorithm
	exactly or to within a factor of $(1+\eps)$~\cite{BarbosaENW16,KumarMVV13,HarveyLL18}. Obtaining MapReduce algorithms that can (almost) match the performance of the greedy algorithm without
	blowing up the number of rounds (in the context of submodular maximization) has been posed as an open question very recently~\cite{LiuV19}. 
	To our knowledge, ours is the \emph{first} parallel implementation of the greedy algorithm in a minimal number of rounds with (almost) no blow-up in the approximation ratio. 
	It is a fascinating open question if our improvement for MWM can be extended to other greedy algorithms, in particular, for constrained submodular maximization. 
\end{itemize}

\vspace{-7pt}

In addition to aforementioned theoretical improvements over previous works, our algorithm has the benefit of being extremely simple (with nearly all the details ``pushed'' to the analysis), making it easily implementable
in the MapReduce model (the only other MapReduce algorithm for matching that we know of to be implemented previously is~\cite{BehnezhadDETY17} which is \emph{limited to bipartite graphs in a crucial way}). We believe this additional feature of our algorithm 
is an important contribution of this paper. We discuss this further in Section~\ref{sec:exp} where we present our experimental results and give an empirical comparison of our algorithm with prior algorithms.

%%% Local Variables:
%%% TeX-master: "main"
%%% End:

%% file: prelim.tex
\section{Preliminaries}\label{sec:prelim}
 Throughout, $[t] := \set{1,\ldots,t}$. For a graph $G(V,E)$, $\opt(G)$ denotes the weight of a maximum weight matching in $G$. 
%For any vertex $v \in V$ and any matching $M \subseteq E$, $M(v)$ denotes the neighbor
%of $v$ in the matching $M$. 

\paragraph{Greedy algorithm.} Let $G(V,E)$ be a graph and $\pi :=\pi(E)$ denote any permutation of $E$. $\Greedy(G,\pi)$ denotes the standard greedy algorithm
which iterates over edges according to $\pi$ and add $e =(u,v)$ to the matching iff both $u$ and $v$ are unmatched. 
It is a standard fact that when $\pi$ is sorted in non-increasing order of weights, $\Greedy(G,\pi)$ outputs a $2$-approximation to $\opt(G)$.

\ssection{MapReduce framework.}
We adopt the MapReduce model as formalized by Karloff~\etal~\cite{KarloffSV10}; see also~\cite{GoodrichSZ11,BeameKS13}. Let $G(V,E)$ with $n:= \card{V}$ and $m := \card{E}$ be the input graph. 
In this model, there are $p$ machines, each with a memory of $s$ such that $p \cdot s = O(m)$, i.e., at most a constant factor larger than the input size, and both $p,s = m^{1-\Omega(1)}$, i.e., 
sublinear in the input size.  The motivation behind these constraints is that the number of machines, and local memory of each machine should be much smaller than the input size to the problem since these frameworks are used to process massive datasets. 
Computation in this model proceeds in synchronous rounds: in each round, each machine performs some local computation
and at the end of the round machines exchange messages to guide the computation for the next round. All messages received by each machine in one round have to fit into the memory of the machine.

%%In this paper, we focus on the \emph{dense} graph regime, i.e., when the number of edges in the input is at least $n^{1+\Omega(1)}$, and hence we always allow $\Omega(n)$ memory per-machine (and still satisfy the requirement of the model on per-machine memory size); in particular, we require that in the last round, the entire solution 
%%to be stored on a single machine (see, e.g.~\cite{MirrokniZ15,BarbosaENW16} for similar assumptions). 
%%In other words, our algorithms are applicable to graphs with a large number of edges but not very large number of vertices. We note that most previous MapReduce algorithms for graph problems also target
%%this regime~\cite{LattanziMSV11,AhnG15,BehnezhadDETY17,AssadiKL17}\footnote{On this note, see also the paper of Leskovec, Kleinberg and Faloutsos~\cite{LeskovecKF2005} that suggests that massive graphs in practice such as citation graphs, affiliation graphs, instant messenger
%%graphs, etc. tend to become denser with $n^{1+\Omega(1)}$ edges over time.}.   
%%

%% file: coreset.tex
\subsection{Randomized Composable Coresets}\label{sec:coreset} 

We briefly review the notion of randomized composable coresets originally introduced by~\cite{MirrokniZ15} in the context of submodular maximization, and further refined in~\cite{AssadiK17} for 
graph problems. Our definition slightly deviates from previous works as we will remark below.

Let $E$ be an edge-set of a weighted graph $G(V,E)$. 
Let $c \geq 1$ be a parameter. A collection of edges $\set{\Ei{1},\ldots,\Ei{k}}$ is a \emph{random $k$-clustering} of $E$ with \emph{expected multiplicity} $c$ iff each edge $e$ in $E$ 
is sent to $c_{e}$ different sets $\Ei{i_1},\ldots,\Ei{i_{c_e}}$ chosen uniformly at random, where $c_e$ is chosen independently for each edge from the \emph{binomial distribution} with $k$ trials and expected value $c$. 
A random clustering of $E$ naturally defines clustering the graph $G$ into $k$ subgraphs $\Gi{1},\ldots,\Gi{k}$ where $\Gi{i} := G(V,\Ei{i})$ for all $i \in [k]$; as a result, 
we use random clustering for both the edge-set and the input graph interchangeably.  
\begin{Definition}[cf.~\cite{MirrokniZ15,AssadiK17}]\label{def:randomized-coreset}
	Consider an algorithm $\alg$ that given a graph $G(V,E)$ outputs a subgraph $\alg(G) \subseteq G$ with at most $s$ edges. 
	Let $k,c \geq 1$ be integers and $\Gi{1},\ldots,\Gi{k}$ denote a random $k$-clustering with expected multiplicity $c$ of a graph $G$. 
	We say that $\alg$ outputs an $\alpha$-approximate $(k,c)$-randomized composable coreset of size $s$ for the weighted matching problem iff
	\begin{align*}
		\alpha \cdot \Ex\bracket{\opt\paren{\alg(\Gi{1}) \cup \ldots \cup \alg(\Gi{k})}} \geq \opt(G),
	\end{align*}
	where $\opt(\cdot)$ denotes the weight of a maximum weight matching in the given graph. Here, the expectation is taken over the random choice of the random clustering.
\end{Definition}

We remark that our definition is somewhat different from~\cite{MirrokniZ15,AssadiK17} in the following sense: Previous works considered the case where each edge
is sent to \emph{exactly} $c$ different subgraphs (only $c=1$ was considered in~\cite{AssadiK17}), while we send each edge to $c$ subgraphs in \emph{expectation}. 
This way of partitioning has the simple yet helpful property that makes the distribution of graphs $\Gi{1},\ldots,\Gi{k}$ a \emph{product} distribution, i.e., each graph $\Gi{i}$ is chosen \emph{independently} even conditioned on all other
graphs in the random $k$-clustering. At the same time, size of each subgraph and the total number of edges across all subgraphs, are still respectively $O(m \cdot c/k)$ and $c\cdot m \pm \Theta(\sqrt{c \cdot m \log{n}})$ with high probability. 

\ssection{Randomized Coresets in MapReduce Framework.} Suppose $G(V,E)$ is the input and let $k:= \sqrt{m \cdot c/s}$. We use a randomized coreset to obtain a MapReduce algorithm: % on $O(k)$ machines as follows: 
\begin{tbox}
\vspace{-5pt}
\begin{enumerate}[leftmargin=15pt]
	\item \textbf{Random clustering:} Create a random $k$-clustering $\Gi{1},\ldots,\Gi{k}$ of expected multiplicity $c$ and allocate each graph $\Gi{i}$ to the machine $i \in [k]$.  
	
	\item \textbf{Coreset:} Each machine $i \in [k]$ creates a randomized composable coreset $C_i \leftarrow \alg(\Gi{i})$. 
	
	\item \textbf{Post-processing:} Collect the union of coresets to create $H:= H(V,C_1,\ldots,C_k)$ on one machine and return a $\beta$-approximation to MWM  on $H$ using \emph{any} offline algorithm. 
\end{enumerate}
\end{tbox}

It is easy to verify that the algorithm requires $O(k) = O(\sqrt{m \cdot c/s})$ machines with $O(\sqrt{m \cdot c \cdot s} + n)$ memory and only \emph{two} rounds of computation. Moreover, by Definition~\ref{def:randomized-coreset}, the 
output of this algorithm is an $(\alpha \cdot \beta)$-approximation to maximum weight matching of $G$. % in expectation. 

%%As we can solve the maximum weight matching problem in polynomial time by an offline algorithm, (or achieve a $(1+\eps)$-approximation in linear time), we can assume $\beta = 1$ (or $\beta = 1+\eps$) and hence solely focus
%%on designing an $\alpha$-approximation randomized coreset which immediately implies the desired MapReduce algorithm. This is the content of next section. 

%% file: weighted-coreset.tex
\section{A Randomized Coreset for Maximum Weight Matching}\label{sec:weighted-matching}

We present a simple randomized composable coreset for MWM in this section and then use it to design an efficient MapReduce algorithm for this problem. In Appendix~\ref{app:coreset-lower}, 
we prove the optimality of the size of our corset using an adaptation of the argument in~\cite{AssadiK17}.% who considered \emph{exact} multiplicity of $1$).

\begin{theorem}\label{thm:weighted-matching-0.5}
	For $\eps > 0$, there exists a $\paren{2+\eps}$-randomized composable coreset of size $O(n)$ with expected multiplicity $O\paren{\frac{\log{(1/\eps)}}{\eps}}$ for the maximum weight matching problem. 
\end{theorem}

Our coreset in Theorem~\ref{thm:weighted-matching-0.5} is simply the greedy algorithm for MWM (with consistent tie-breaking). Let $G(V, E)$ be a graph and $\Gi{1},\ldots,\Gi{k}$ be a random $k$-clustering of $G$ with expected multiplicity $c$. 
We propose the
$\GreedyCoreset$ for approximating MWM on $G$: on each subgraph $\Gi{i}$, simply return $M_i:= \Greedy(\Gi{i},\pi_i)$ as the coreset, where $\pi_i$ sorts 
the edges in $\Gi{i}$ in non-increasing order of their weights (breaking the ties \emph{consistently} across all $i \in [k]$). In the following lemma, we analyze the performance of this coreset.

\begin{lemma}\label{lem:weighted-matching}
	Suppose $\Gi{1},\ldots,\Gi{k}$ is a random $k$-clustering of $G$ with expected multiplicity $c$ and $M_i := \Greedy(\Gi{i},\pi_i)$.
	Define the graph $H(V,E(H))$ with $E(H):= \bigcup_{i=1}^{k} M_i$; then, 
	\[\Ex\bracket{\opt(H)} \geq \paren{\frac{1}{2}-O\paren{\frac{\log{c}}{c}}} \cdot \opt(G).\] 
\end{lemma}
\noindent
Theorem~\ref{thm:weighted-matching-0.5} follows immediately from Lemma~\ref{lem:weighted-matching} (by setting $c = \Theta\paren{\frac{\log{(1/\eps)}}{\eps}}$). The rest of this section is devoted to the proof of Lemma~\ref{lem:weighted-matching}. 
	
	\ssection{Notation.} Let $\pi$ be a permutation of all edges in $G$ in non-increasing order of their weights (\emph{consistent} with orderings $\pi_i$ for $i \in [k]$). Notice that for each $i \in [k]$, $M_i = \Greedy(\Gi{i},\pi)$ as well. Hence, in the following, 
	we use the permutation $\pi$ instead of each $\pi_i$. Throughout the proof, we fix an arbitrary maximum weight matching $\Mstar$ in 
	$G$ and denote by $\ww{\Mstar} = \opt(G)$ the weight of $\Mstar$. For any edge $e \in \Mstar$, we use $\pi^{<e}$ to refer to the set of edges in $\pi$ that appear before $e$. We slightly abuse the notation and 
	use $\Greedy(G,\pi^{<e})$ to mean that we run $\Greedy(G,\pi)$ and stop exactly before processing the edge $e$, i.e., we only consider the edges in $\pi^{<e}$. %The following definition is crucial in our analysis. 
	
	\begin{Definition}[Free/Blocked Edges]\label{def:free-blocked}
	We say that an edge $e \in \Mstar$ is \emph{free} for machine $i \in [k]$ iff no end point of $e$ is matched by $\Greedy(\Gi{i},\pi^{<e})$; otherwise we call $e \in \Mstar$ \emph{blocked}. We use $F_i$ to denote the 
	set of free edges in $\Gi{i}$ and $B_i$ to denote the blocked edges. 
	\end{Definition}
	
	We emphasize that in Definition~\ref{def:free-blocked}, an edge $e \in \Mstar$ can be free or blocked on some machine $i \in [k]$, without necessarily even appearing in $\Gi{i}$. In other words, this definition is independent of whether 
	$e$ belongs to $\Gi{i}$ or not. However, notice that if an edge $e$ is free on machine $i$ and it also appears in $\Gi{i}$, then $e$ would definitely belong to the matching $M_i$ (i.e., the coreset on machine $i$). On the other hand, 
	if an edge $e$ is blocked in machine $i$, then necessarily some edge $e'$ exists in $M_i$ such that $e'$ is incident on $e$ and $\ww{e'} \geq \ww{e}$. We refer to $e'$ as the \emph{certificate} of $e$ in machine $i$.

	\ssection{Overview.} The idea behind the proof of Lemma~\ref{lem:weighted-matching} is as follows. Recall that the distribution of each graph $\Gi{i}$ in the random clustering is the same, and is independent of other graphs. Hence, we can focus on each 
	machine $i \in [k]$, say machine $1$, separately. Consider blocked edges $B_1$ in machine $1$: for any such edge, we have already picked another edge with at least the same weight in the matching $M_1$ of machine $1$ (by definition of an edge being 
	blocked). 	
	We can hence use a simple charging argument here to argue that the matching $M_1$ of machine $1$ already has enough edges to ``compensate'' for blocked edges of $\Mstar$ that were not picked by machine $1$. \
	
	The main part of the argument is 
	however to show that we can find enough edges in $M_2,\ldots,M_k$ chosen by other machines that can be added to $M_1$ to also compensate for free edges in machine $1$. The idea here is that since the distribution of input to all machines
	is the same and is independent across, if an edge $e$ is free in machine $1$, it is ``most likely'' free in \emph{many} other machines as well, in particular, in a machine $j \in [k] \setminus \set{1}$ which also \emph{contains} this edge. By definition, this edge
	then would be chosen in matching $M_j$. We then use another careful charging scheme to argue that we can indeed ``augment'' the matching $M_1$ by free edges in $F_1$ that appear in $M_2,\ldots,M_k$ to obtain an almost-two approximation. The main 
	difficulty here is that 
	even though edges in $F_1$ were free in machine $1$, they may still be incident on edges in $M_1$ with equal or smaller weight (as being free only implies that these edges were not incident on
	 edges with \emph{higher} weight) and hence they cannot be readily 
	added to $M_1$; this is the reason we need to find ``short augmenting paths'' in $F_1 \cup M_1$ which may require switching some edges out of $M_1$ as well. 
	
	\smallskip
	
	We now start with the formal proof. Throughout, define $F'_1 := F_1 \cap E(H)$, i.e., the set of free edges in machine $1$ that are present in $H$. We refer to these
	edges as \emph{available} free edges (we prove later that essentially any free edge is also available with a large probability). %this set is almost identical to $F_1$ in a probabilistic sense). . 
	To perform the charging, we need to partition the edges of $F'_1,B_1$ and $M_1$ as follows (see Figure~\ref{fig:types} for an illustration):  
	
\begin{figure*}[h!]
    \centering
    \subcaptionbox{A set of type $1$ edges.}[0.3\textwidth]{

    \begin{tikzpicture}[ auto ,node distance =1cm and 2cm , on grid , semithick , state/.style ={ circle ,top color =white , bottom color = white , draw, black , text=black}, every node/.style={inner sep=0,outer sep=0}]

\node[state, circle, black, line width=0.35mm, minimum height=7pt, minimum width=7pt] (a1){};
\node[state, circle, black, line width=0.35mm, minimum height=7pt, minimum width=7pt] (a2) [below =1.5cm of a1]{};
\node[state, circle, black, line width=0.35mm, minimum height=7pt, minimum width=7pt] (a3) [right=3cm of a1]{};
\node[state, circle, black, line width=0.35mm, minimum height=7pt, minimum width=7pt] (a4) [below=1.5cm of a3]{};

\draw[line width=0.5mm, red](a1) -- node[above=7pt of a1]{\textcolor{black}{$f$}} (a3);
\draw[line width=0.5mm, red] (a2) -- node[below=7pt of a2]{\textcolor{black}{$e$}} (a4);
\draw[line width=0.2mm] (a1) -- node [below=7pt of a1]{$e'$} (a4);

\end{tikzpicture}

}
\hspace{0.2cm}   
    \subcaptionbox{A set of type $2$ edges.}[0.3\textwidth]{

    \begin{tikzpicture}[auto ,node distance =1cm and 2cm , on grid , semithick , state/.style ={ circle ,top color =white , bottom color = white , draw, black , text=black}, every node/.style={inner sep=0,outer sep=0}]

\node[state, circle, black, line width=0.35mm, minimum height=7pt, minimum width=7pt] (a1){};
\node[state, circle, black, line width=0.35mm, minimum height=7pt, minimum width=7pt] (a2) [below =1.5cm of a1]{};
\node[state, circle, black, line width=0.35mm, minimum height=7pt, minimum width=7pt] (a3) [right=3cm of a1]{};
\node[state, circle, black, line width=0.35mm, minimum height=7pt, minimum width=7pt] (a4) [below=1.5cm of a3]{};

\draw[line width=0.5mm, ForestGreen, dashed](a1) -- node[above=7pt of a1]{\textcolor{black}{$f$}} (a3);
\draw[line width=0.5mm, red] (a2) -- node[below=7pt of a2]{\textcolor{black}{$e$}} (a4);
\draw[line width=0.2mm] (a1) -- node [below=7pt of a1]{$e'$} (a4);

\end{tikzpicture}

}\hspace{0.2cm}
    \subcaptionbox{A set of type $3$ edges.}[0.3\textwidth]{

    \begin{tikzpicture}[auto ,node distance =1cm and 2cm , on grid , semithick , state/.style ={ circle ,top color =white , bottom color = white , draw, black , text=black}, every node/.style={inner sep=0,outer sep=0}]

\node[state, circle, black, line width=0.35mm, minimum height=7pt, minimum width=7pt] (a1){};
\node[state, circle, black, line width=0.35mm, minimum height=7pt, minimum width=7pt] (a2) [below =1.5cm of a1]{};
\node[state, circle, black, line width=0.35mm, minimum height=7pt, minimum width=7pt] (a3) [right=3cm of a1]{};
\node[state, circle, black, line width=0.35mm, minimum height=7pt, minimum width=7pt] (a4) [below=1.5cm of a3]{};
\node[state, circle, black, line width=0.35mm, minimum height=7pt, minimum width=7pt] (a5) [above=1.5cm of a1]{};
\node[state, circle, black, line width=0.35mm, minimum height=7pt, minimum width=7pt] (a6) [right=3cm of a5]{};

\draw[line width=0.5mm, red] (a5) -- node[above=7pt of a5]{\textcolor{black}{$e$}} (a6);
\draw[line width=0.2mm] (a5) -- node [above=9pt of a1]{$e''$} (a3);
\draw[line width=0.5mm, ForestGreen, dashed](a1) -- node[above=7pt of a1]{\textcolor{black}{$f$}} (a3);
\draw[line width=0.5mm, red] (a2) -- node[below=7pt of a2]{\textcolor{black}{$e$}} (a4);
\draw[line width=0.2mm] (a1) -- node [below=7pt of a1]{$e'$} (a4);

\end{tikzpicture}

}
    
    \caption{Illustration of the partitioning used in Lemma~\ref{lem:weighted-matching}. Thick solid edges (red) are blocked edges, thick dashed edges (green) are available free edges, and normal edges (black) are certificate edges.}
    \label{fig:types}
\end{figure*}
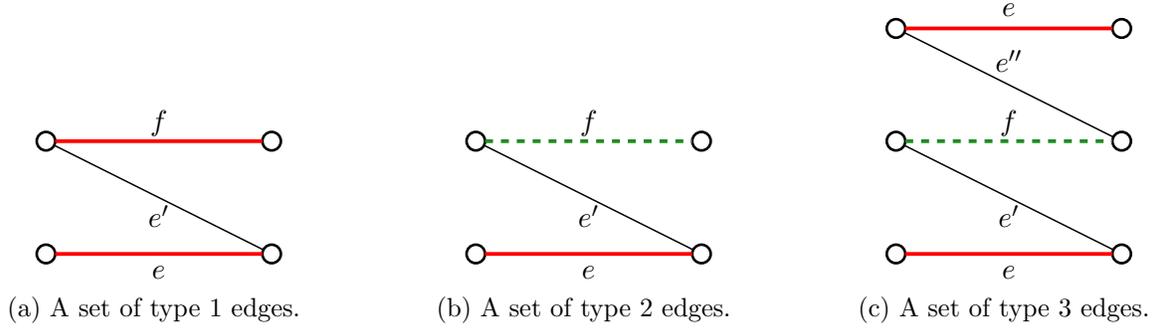

	\begin{enumerate}[leftmargin=20pt]
		\item Let $e$ be an edge with maximum weight in $B_1$ and $e'$ be its certificate in $M_1$.  Additionally, let $f$ be the \emph{other} edge incident on $e'$ in $F'_1 \cup B_1$ (set $f =\: \perp$ if no such edge exists). 
		\item \textbf{Type $\mathbf{1}$ edges:} If $f$ belongs to $B_1$ or is $\perp$, then add $e'$ to $M_{1,1}$ and $e ,f$ to $B_{1,1}$. Remove \emph{both} $e$ and $f$ from $B_1$. We refer to $(e,f,e')$ as a set of type $1$ edges. 
		\item \textbf{Type $\mathbf{2}$ edges:} If $f \in F'_1$ and $f$ is \emph{not} incident on any certificate edge other than $e'$, then add $f$ to $F'_{1,2}$, $e$ to $B_{1,2}$, and $e'$ to $M_{1,2}$. Remove $f$ from $F'_1$ and $e$ from $B_1$. 
		We refer to $(e,f,e')$ as a set of type $2$ edges. 
		\item \textbf{Type $\mathbf{3}$ edges:} If $f \in F'_1$ and $f$ is incident on another certificate edge $e''$ which is a certificate for some edge $z \in B_1$, then add $f$ to $F'_{1,3}$, $e,z$ to $B_{1,3}$, and $e',e''$ to $M_{1,3}$. Remove 
		$f$ from $F'_{1,3}$ and $e,z$ from $B_{1,3}$. We refer to $(e,f,z,e',e'')$ as a set of type $3$ edges. 
		\item Continue the process from first line until no edge remains in $B_1$. 	Add the remaining edges in $F'_1$ after this step the set to $F'_{1,0}$. 
	\end{enumerate}
	\noindent
	It is immediate to verify that $F'_1 = F'_{1,0} \cup F'_{1,2} \cup F'_{1,3}$, $B_1 = B_{1,1} \cup B_{1,2} \cup B_{1,3}$ and $M_{1,1} \cup M_{1,2} \cup M_{1,3} \subseteq M_1$, and all these sets are pairwise disjoint. 
	The following claim states  properties of this partitioning and the charging scheme. 

	\begin{claim}\label{clm:types}
		In the partitioning scheme, for any set of edges: 
		\begin{itemize}[leftmargin=10pt]
		\item type $1$ $(e,f,e')$: $\ww{e'} \geq \frac{1}{2} \cdot \paren{\ww{e} + \ww{f}}$. 
		\item type $2$ $(e,f,e')$: $\ww{f} \geq \ww{e'} \geq \ww{e}$. 
		\item type $3$ $(e,f,z,e',e'')$: $\max\set{\ww{f}, \ww{e'} + \ww{e''}} \geq \frac{1}{2} \cdot \paren{\ww{e}+\ww{f}+\ww{z}}$. 
		\end{itemize}
	\end{claim}
	\begin{proof}
		\begin{itemize}[leftmargin=20pt]
			\item Since $e$ is chosen before $f$, $\ww{e} \geq \ww{f}$. Also, since $e'$ is a certificate of $e$, $\ww{e'} \geq \ww{e}$, a simple calculation finalizes the proof. 
			\item Since $f$ is free, it means $e'$ is chosen after visiting $f$; hence $\ww{f} \geq \ww{e'}$. Since $e'$ is a certificate for $e$, $\ww{e'} \geq \ww{e}$. 
			\item Follows since $e'$ and $e''$ are certificates for $e$ and $z$, respectively, and hence $\ww{e'} \geq \ww{e}$ and $\ww{e''} \geq \ww{z}$. 
		\end{itemize}
		This finalizes the proof.
	\end{proof}	
	We now lower bound the weight of the maximum weight matching in the graph $H$. 
	
	\begin{lemma}\label{lem:wm-charging}
		$\opt(H) \geq \frac{1}{2} \cdot \paren{\ww{F'_1} + \ww{B_1}}$. 
	\end{lemma}
	\begin{proof}
		
		Define $M' := F'_1 \cup B_1 \subseteq \Mstar$, and hence $\ww{M'} = \ww{F'_1} + \ww{B_1}$. 
		To prove the lemma, it suffices to show that there exists a matching $M$ in $H$ with $w(M) \geq \frac{1}{2} \cdot w(M')$. Notice that edges in $F'_1$ ($\subseteq M'$) are already in $H$ and hence can be used in $M$. 
		Moreover, for any edge $e \in B_1$, a certificate of $e$ also belongs to the matching $M_1$. We now show how to design the matching $M$ using only the certificate edges and edges in $F'_1$ such
		that $\ww{M} \geq \frac{1}{2} \cdot \ww{M'_1}$. We use the partitioning scheme in Claim~\ref{clm:types} to construct $M$ as follows: 
		\begin{enumerate}[leftmargin=20pt]
			\item For a set of type $1$ edges $(e,f,e')$: add $e' \in M_1$ to the matching $M$. 
			\item For a set of type $2$ edges $(e,f,e')$: add $f \in F'_1$ to $M$.
			\item For a set of type $3$ edges $(e,f,z,e',e'')$: if $\ww{f} \geq \ww{e'}+\ww{e''}$, add $f \in F'_1$ to $M$, otherwise add $e',e'' \in M_1$ to $M$.
			\item Add any edge in $F'_{1,0}$ to $M$. 
		\end{enumerate}
		\noindent
		It is immediate to verify that all edges added to $M$ indeed belong to $H$ and they also form a matching. 
		To lower bound the weight of $M$, notice that by Claim~\ref{clm:types}, for any types of edges in $M' \cup M_1$, at least half the weights of edges in $M'$ also appear in $M$. This immediately implies that $\ww{M} \geq \frac{1}{2} \cdot \ww{M'}$,
		finalizing the proof. 
	 \end{proof}
	
	We now use the randomness in the clustering to argue that nearly all free edges on machine $1$ are also available. 
	
	\begin{lemma}\label{lem:wm-prob}
		$\Ex\Bracket{\ww{F'_1}} \geq \Ex\Bracket{\ww{F_1}} - O\paren{\frac{\log{c}}{c}} \cdot \opt.$ 
	\end{lemma}
	\begin{proof}
		We first have, 
		\begin{align*}
			\Ex\Bracket{\ww{F'_1}} &= \sum_{e \in \Mstar} \ww{e} \cdot \Pr\paren{\text{$e \in F_1 \wedge e \in M_j$ for some $j \in [k]$}}\\ 
			&\geq \sum_{e \in \Mstar} \ww{e} \cdot \Pr\paren{\text{$e \in F_1$}} \cdot \Pr\paren{\text{$e \in M_j$ for some $j \in [k]$}}
		\end{align*}
		This is because random clustering induces a product distribution on inputs across machines and hence the fact that $e$ is free in $\Gi{1}$ is independent of whether $e$ is picked in some other matching $M_j$ for $j \neq i$. 
		We can now calculate the probability that an edge $e$ belongs to a matching $M_j$ for a fixed $j \in [k]$. 
		\begin{align}
			\Pr\paren{\text{$e \in M_j$}} &= \Pr\paren{\text{$e \in F_j$ and $e$ is sampled in $\Gi{j}$}} \notag \\
			&= \Pr\paren{\text{$e \in F_j$}} \cdot \Pr\paren{\text{$e \in \Gi{j}$}} \tag{the event ``$e \in F_j$'' only depends on the edges in $\pi^{<e}$} \\
			&= \Pr\paren{\text{$e \in F_j$}} \cdot \frac{c}{k} \label{eq:prob-Mj},
		\end{align}
		since each edge appears in $\Gi{j}$ with probability $c/k$. Now notice that the marginal distribution of the graph $\Gi{1}$ and $\Gi{j}$ under random clustering is the same; as a result the probability that an edge is good in $\Gi{j}$ is equal to
		this probability for the graph $\Gi{1}$ as well. Using this, plus the fact that the event that $e$ belongs to $M_j$ is independent of all other graphs $\Gi{\ell}$ for $\ell \neq j$, we have, 
		\begin{align*}
			\Pr\paren{\text{$e \in M_j$ for $j \in [k]\setminus\set{1}$}} &= 1- \prod_{j \in [k] \setminus \set{1}}\Pr\paren{e \notin M_j}  \\
			&=  1- \prod_{j \in [k] \setminus \set{1}}\Paren{1-\Pr\paren{\text{$e \in F_j$}} \cdot \frac{c}{k}} \tag{by Eq~(\ref{eq:prob-Mj})}\\
			&= 1- \Paren{1-\Pr\paren{\text{$e \in F_1$}} \cdot \frac{c}{k}}^{k-1}.
		\end{align*}
		
		Define $S \subseteq \Mstar$ as the set of all edge $e \in \Mstar$ such that $\Pr\paren{\text{$e$ is free in $\Gi{i}$}} \geq \frac{4\log{c}}{c}$. By above equation, for any edge $e \in S$, we have that, 
		\begin{align*}
			\Pr\paren{\text{$e\in M_j$ for $j \in [k]\setminus\set{1}$}} &\geq 1- \Paren{1-\frac{4\log{c}}{c} \cdot \frac{c}{k}}^{k-1} \geq 1 - e^{\paren{-2\log{c}}} = 1-O(1/c). 
		\end{align*}
		Consequently, 
		\begin{align*}
			\Ex\Bracket{\ww{F'_1}} & \geq \sum_{e \in S}  \ww{e} \cdot \Pr\paren{e \in F_1}  \cdot \Pr\paren{\text{$e \in M_j$ for $j \in [k]\setminus\set{1}$}}\\
			&\geq \sum_{e \in S} \ww{e} \cdot \Pr\paren{\text{$e \in F_1$}} \cdot  \paren{1-O(1/c)}\\
			&= \paren{1-O(1/c)} \cdot \\ 
			&\paren{\sum_{e \in \Mstar} \Pr\paren{\text{$e \in F_1$}} \cdot \ww{e} - \sum_{e \in \Mstar\setminus S} \Pr\paren{\text{$e \in G_1$}} \cdot \ww{e}} \\
			&\geq \paren{1-O(1/c)} \cdot \paren{\Ex\bracket{\ww{F_i}} - \opt \cdot O(\frac{\log{c}}{c})} \\
			&\geq \Ex\bracket{\ww{F_i}} - O\paren{\frac{\log{c}}{c}} \cdot \opt,
		\end{align*}
		which finalizes the proof. 
	\end{proof}
		
	Lemma~\ref{lem:weighted-matching} now follows immediately from Lemma~\ref{lem:wm-charging} and Lemma~\ref{lem:wm-prob} as 
	\begin{align*}
		\Ex\bracket{\opt(H)} &\Geq{Lemma~\ref{lem:wm-charging}} \frac{1}{2} \cdot \Ex\Bracket{\ww{F'_1} + \ww{B_1}} \\
		&\Geq{Lemma~\ref{lem:wm-prob}} \frac{1}{2} \cdot \Ex\bracket{\ww{F_1} + \ww{B_1}} - O\paren{\frac{\log{c}}{c}} \cdot \opt, 
	\end{align*}
	since $F_1 \cup B_1 = \Mstar$ and $\ww{\Mstar} = \opt$. This concludes proof of Lemma~\ref{lem:weighted-matching} and Theorem~\ref{thm:weighted-matching-0.5}.
	
\subsection{MapReduce Algorithms for Maximum Weight Matching}\label{sec:wm-mapreduce}

We now present our MapReduce algorithm based on the coreset in Theorem~\ref{thm:weighted-matching-0.5}. 

\begin{theorem}\label{thm:main}
	There is a randomized MapReduce algorithm that for any $\eps > 0$, outputs a $(2+\eps)$-approximation to maximum weight matching in expectation using $O(\sqrt{\frac{m \cdot \log{(1/\eps)}}{\eps \cdot n}})$ 
	machines	each with $O(\sqrt{mn \cdot (1/\eps) \cdot \log{(1/\eps)}} + n)$ memory and in only two rounds of parallel computation. The local computation on each machine requires
	$O((1/\eps) \cdot \log{(1/\eps)} \cdot \sqrt{mn \cdot (1/\eps) \cdot \log{(1/\eps)}} + n)$ time. 
	Here, $m$ and $n$ denote the number of edges and vertices in the input 
	graph, respectively. 
\end{theorem}
\begin{proof}
	We simply use the randomized coreset in Theorem~\ref{thm:weighted-matching-0.5} which has expected multiplicity $c := O((1/\eps) \cdot \log{(1/\eps)})$ with the reduction in Section~\ref{sec:coreset}. The bound on the
	number of rounds, memory per-machine, and the total number of machines follows immediately from this and the discussion in Section~\ref{sec:coreset}. To obtain the bounds on the running time of the algorithm, 
	we simply need to run $\Greedy$ to compute each coreset which takes $O(\sqrt{mn \cdot (1/\eps) \cdot \log{(1/\eps)}} + n)$ time as the total number of edges in the input of each machine is bounded by $O(\sqrt{mn \cdot (1/\eps) \cdot \log{(1/\eps)}})$. 
	To compute the final answer, we run the algorithm of~\cite{DuanP14} which takes $O(m' \cdot (1/\eps) \cdot \log{(1/\eps)} + n)$ time on an $n$-vertex graph with $m'$ edges, where $m' = O(\sqrt{mn \cdot (1/\eps) \cdot \log{(1/\eps)}})$ here again. 
\end{proof}

The algorithm in Theorem~\ref{thm:main} implies a theoretically efficient (almost) $2$-approximation for MWM in the MapReduce model. Implementing this algorithm in practice however can be slightly challenging, simply
due to the post-processing step which requires computing an (almost) maximum weight matching (such algorithms tend to be tricky on general graphs, mainly to handle ``blossoms'' (see, e.g.,~\cite{GabowT91,DuanP14}) although one can
use \emph{any} readily available algorithm for MWM in this step). It is thus natural to consider simpler post-processing steps also that are easier to implement in practice. An obvious candidate here is the greedy algorithm itself. It follows already from the guarantee
of $\Greedy$ and Theorem~\ref{thm:weighted-matching-0.5} (by the same exact argument as in Theorem~\ref{thm:main}) that this would lead to an (almost) $4$-approximation. 
We next prove that this algorithm in fact already achieves an improved approximation  of $3$. % (we also remark on a beyond-worst case analysis of this algorithm in Appendix~\ref{app:greedy-why}).

\begin{theorem}\label{THM:GREEDY}
	The MapReduce algorithm for maximum weight matching obtained by applying $\Greedy$ as the post-processing step to $\GreedyCoreset$ always outputs a $(3+\eps)$-approximation in expectation. 
\end{theorem}

We present the proof of Theorem~\ref{THM:GREEDY} in the next section.

\subsection{Proof of Theorem~\ref{THM:GREEDY}}\label{sec:thm-greedy}

	Proof of this theorem is based on a more careful analysis of $\GreedyCoreset$ and writing a factor-revealing LP for the overall approximation ratio of the algorithm. To be precise, in the post-processing step, we do as follows: we run the
	greedy algorithm on the union of coresets received from all machines and return either this matching, denoted by $M_G$, or simply return any arbitrary matching computing as a coreset on one machine, say $M_1$ on machine $1$, 
	as the final solution, based on which one has a larger weight (we note that returning matching $M_1$ is only used to handle the ``pathological case'' that arise in the analysis in which each graph $\Gi{i}$ already contains a very large weighted
	matching, i.e., even \emph{after} sampling $G$ by a factor of $\Theta(1/k)$, we can still find a near optimum solution in the sampled graph). 

	Recall the definition of (available) free and blocked edges from the proof of Lemma~\ref{lem:weighted-matching} and let $F'_1$, $B_1$, and $M_1$ be defined as before (along with their partitioning into different types). 
	We start with the following simple lemma. 
\begin{lemma}\label{lem:wm-greedy-m1}
	$\ww{M_1} \geq \frac{1}{2} \cdot \ww{B_{1,1}} + \ww{B_{1,2}} + \ww{B_{1,3}}$. 
\end{lemma}
\begin{proof}
	For any pair of edges $e,f$ in $B_{1,1}$ that belong to the same type $1$ set, there exist a unique edge $e'$ in $M_{1,1}$ with $\ww{e'} \geq \frac{1}{2} \cdot \paren{\ww{e}+\ww{f}}$
	by first part of Claim~\ref{clm:types}. Hence, $\ww{M_{1,1}} \geq \frac{1}{2} \cdot \ww{B_{1,1}}$. 
	
	For any edge $e$ in $B_{1,2}$ that belongs to a type $2$ set, there exists a unique edge $e'$ in $M_{1,2}$ with $\ww{e'} \geq \ww{e}$ by
	second part of Claim~\ref{clm:types}. Hence, $\ww{M_{1,2}} \geq \ww{B_{1,2}}$. 
	
	For any pair of edges $e,z$ in $B_{1,3}$ that belong to the same type $3$ set, there exists a unique pair of edges $e,e'$ in $M_{1,3}$ with $\ww{e} + \ww{e'} \geq \ww{e} + \ww{z}$ by (the proof of) third part of Claim~\ref{clm:types}.
	Hence, $\ww{M_{1,3}} \geq \ww{B_{1,3}}$. 
	
	Finally, since $M_{1,1},M_{1,2},M_{1,3}$ are disjoint and are all subset of $M_1$, we have $\ww{M_1} \geq \ww{M_{1,1}} + \ww{M_{1,2}} + \ww{M_{1,3}} \geq  \frac{1}{2} \cdot \ww{B_{1,1}} + \ww{B_{1,2}} + \ww{B_{1,3}}$, as desired. 
\end{proof}

Recall that $H:=H(V,M_1 \cup \ldots \cup M_k)$ is the graph consists of all edges contained in the $k$ coresets. The following lemma is stronger version of Lemma~\ref{lem:wm-charging} from previous section. 

\begin{lemma}\label{lem:wm-charging-stronger}
	$\opt(H) \geq \ww{F'_{1,0}} + \frac{1}{2} \cdot \ww{B_{1,1}} + \ww{F'_{1,2}} + \max\set{\ww{F'_{1,3}} , \ww{B_{1,3}}}$. 
\end{lemma}
\begin{proof}
	We compute a matching $M$ inside the subgraph $H$ as in Lemma~\ref{lem:wm-charging-stronger}. Edges in $F'_{1,0}$ belong to $M$ by construction. For any pair of edges $(e,f)$ in $B_{1,1}$, there exists a unique edge $e' \in M$ (chosen from $M_1$) with
	$\ww{e'} \geq \frac{1}{2} \cdot \paren{\ww{e} + \ww{f}}$ by Claim~\ref{clm:types}. Any edge in $F'_{1,2}$ also belongs to $M$ by construction. Finally, for any triple of type $3$ edges $(e,f,z)$, the weight of the edges added to $M$ by construction (using third 
	part of Claim~\ref{clm:types}) is at least equal to $\max{\ww{f},\ww{e}+\ww{z}}$, hence clearly these edges have weight at least $\max\set{\ww{F'_{1,3}} + \ww{B_{1,3}}}$ in total (note that in general, one can even get a higher value here by ``mix and 
	matching'' the max-term for each triple of edges separately, but we shall not do that here, as our analysis suggests this further extension is not helpful). 
\end{proof}

%%We also have an analogue of Lemma~\ref{lem:wm-prob} for any of the sets $F'_{1,j}$ for $j \in \set{0,2,3}$. 
%%
%%\begin{lemma}\label{lem:wm-prob-stronger}
%%	For any $j \in \set{0,2,3}$, $\Ex\Bracket{\ww{F'_{1,j}}} \geq \Ex\Bracket{\ww{F_{1,j}}} - O\paren{\frac{\log{c}}{c}} \cdot \opt$.
%%\end{lemma}
%%\begin{proof}
%%The proof is identical to the proof of Lemma~\ref{lem:wm-prob}.
%%\end{proof}

We also have the following simple relation between the weight of blocked and available free edges (by second and third part of Claim~\ref{clm:types}):
\begin{align}
	\ww{F'_{1,2}} \geq \ww{B_{1,2}}, \\
	\ww{F'_{1,3}} \geq \frac{1}{2} \cdot \ww{B_{1,3}} \label{eq:f12-b12}.
\end{align}
\noindent
Let $M_G$ denote the final matching computed by the greedy algorithm on the graph $H$. As greedy always outputs a $2$-approximation of $\opt(H)$, we have: 
\begin{align}
	\ww{M_G} &\geq \frac{1}{2} \cdot \opt(H) \label{eq:greedy-opt}.
\end{align}
Finally define the following quantity: 
\begin{align}
	O &:= \Ex\bracket{\ww{F'_{1,0}}} + \Ex\bracket{\ww{F'_{1,2}}}  + \Ex\bracket{\ww{F'_{1,3}}} 
	+ \Ex\bracket{\ww{B_{1,1}}}  + \Ex\bracket{\ww{B_{1,2}}}  + \Ex\bracket{\ww{B_{1,3}}}. \label{eq:O} 
\end{align}
Note that by Lemma~\ref{lem:wm-prob} in previous section and choice of $c$ in $\GreedyCoreset$, 
\begin{align}
	O \geq (1-\eps) \cdot \opt(G). \label{eq:opt-opt} 
\end{align}

We now use the equations in Lemmas~\ref{lem:wm-greedy-m1},~\ref{lem:wm-charging-stronger}, together with Eq~(\ref{eq:f12-b12}), Eq~(\ref{eq:greedy-opt}), and Eq~(\ref{eq:opt-opt}), to form a factor-revealing LP for the performance of the 
algorithm in Theorem~\ref{THM:GREEDY}. Define the following variables (they are \emph{scaled by $O$} and are hence between $0$ and $1$): 
\begin{align*}
	\alpha_{0} &:= \frac{1}{O} \cdot \Ex\bracket{\ww{F'_{1,0}}}, \qquad \alpha_{2} := \frac{1}{O} \cdot\Ex\bracket{\ww{F'_{1,2}}} 
	\qquad \alpha_{3} := \frac{1}{O} \cdot\Ex\bracket{\ww{F'_{1,3}}}, \\
	\beta_{1} &:= \frac{1}{O} \cdot\Ex\bracket{\ww{B_{1,1}}}, \qquad \beta_{2}~:= \frac{1}{O} \cdot\Ex\bracket{\ww{B_{1,2}}} 
	\qquad \beta_{3} := \frac{1}{O} \cdot\Ex\bracket{\ww{B_{1,3}}}, \\
	Z &:= \textnormal{size of output matching scaled by $O$.} 
\end{align*}
\noindent
We design the following linear program: 
\begin{tbox}
	A factor-revealing LP for the performance of the algorithm in Theorem~\ref{THM:GREEDY}. 
	\begin{align}
		&\textbf{Minimize} \qquad Z \notag \\ 
		&\textbf{Subject to.} \quad 1 = \alpha_0 + \alpha_2 + \alpha_3 + \beta_1 + \beta_2 + \beta_3 \label{const:1} \\
		&\qquad \qquad \qquad~~Z \geq \frac{1}{2} \cdot \beta_1 + \beta_2 + \beta_3 \label{const:2} \\
		&\qquad \qquad \qquad~~Z \geq \frac{1}{2} \cdot \Paren{\alpha_0 + \frac{1}{2} \cdot \beta_1 + \alpha_2 +  \alpha_3} \label{const:3} \\
		&\qquad \qquad \qquad~~Z \geq \frac{1}{2} \cdot \Paren{\alpha_0 + \frac{1}{2} \cdot \beta_1 + \alpha_2 +  \beta_3} \label{const:4} \\
		&\qquad \qquad \qquad~~\alpha_2 \geq \beta_2 \label{const:5} \\
		&\qquad \qquad \qquad~~\alpha_3 \geq \frac{1}{2} \cdot \beta_3 \label{const:6} \\
		&\qquad \qquad \qquad ~~\alpha_0,\alpha_2,\alpha_3,\beta_1,\beta_2,\beta_3 \in [0,1] \label{const:7}.
	\end{align}
\end{tbox}
\noindent
We have the following lemma on the correctness of this factor-revealing LP.
\begin{lemma}\label{lem:factor-revealing}
	Let $M$ be the matching output by the algorithm in Theorem~\ref{THM:GREEDY} and $Z^*$ be the optimal value of LP; then, $\Ex\bracket{\ww{M}} \geq Z^* \cdot (1-\eps) \cdot \opt(G)$. 
\end{lemma}
\begin{proof}
	Recall that variables in the LP are in fact scaled by the value of $\opt(G)$ and are hence in $[0,1]$. We now have: 
	\begin{itemize}[leftmargin=20pt]
		\item Constraint~(\ref{const:1}) holds by definition of $O$ in Eq~(\ref{eq:O}). 
		\item Constraint~(\ref{const:2}) follows from Lemma~\ref{lem:wm-greedy-m1} as the algorithm can return the matching $M_1$ as solution (in case its weight is larger than the greedy solution $M_G$ computed). 
		\item Constraint~(\ref{const:3}) follows from Lemma~\ref{lem:wm-charging-stronger} when the algorithm returns the matching $M_G$ (considering the first term in ``max'').  
		\item Constraint~(\ref{const:4}) follows from Lemma~\ref{lem:wm-charging-stronger} when the algorithm returns the matching $M_G$ (considering the second term in ``max'').
		\item Constraint~(\ref{const:5}) holds by the first term in Eq~(\ref{eq:f12-b12}). 
		\item Constraint~(\ref{const:6}) holds by the second term in Eq~(\ref{eq:f12-b12}). 
		\item Constraint~(\ref{const:7}) holds as we are rescaling every variable by $O$. 
	\end{itemize}
	Finally, the solution of the LP is always at least as large as both matchings $M_1$ and $M_G$ and is stated relative to $O$ which by Eq~(\ref{eq:opt-opt}) is at least $(1-\eps) \cdot \opt(G)$, hence finalizing the proof. 
\end{proof}

As this LP contains only a constant number of variables, we can solve it exactly using any standard LP-solver such as Gurobi's command line tool~\cite{gurobi} which results in the following claim. 
%(a simple code using Gurobi's command line tool~\cite{gurobi} is provided in Appendix~\ref{app:gurobi}). 

\begin{claim}\label{clm:LP-value}
	The optimal solution $Z^*$ of LP is  $Z^* = 1/3$. 
\end{claim}
\noindent
Theorem~\ref{THM:GREEDY} now follows immediately from Lemma~\ref{lem:factor-revealing} and Claim~\ref{clm:LP-value}. 

\begin{remark}\label{rem:too-good}
	\emph{
	Recall that in the analysis of Theorem~\ref{THM:GREEDY}, $M_1$ denotes the weight of the greedy solution on a single subgraph, i.e., after subsampling the graph $G$ by a 
factor of $c/k$. It is thus natural to expect that weight of of $M_1$ should be much smaller than $\opt(G)$. After all, any single weighted matching in $G$ would incur a ``loss'' of factor $c/k$ in its weight in $\Gi{1}$ due to sampling, and unless 
$G$ contains ``many'' (at least $\Omega(k/c)$) \emph{edge-disjoint near-optimal weighted matchings}, then indeed $\ww{M_1} \ll \opt(G)$. 
This suggests that in most cases, we can expect weight of $M_1$ to be much smaller than $\opt(G)$ (we demonstrate this on the datasets studied in this paper in our experimental studies in Section~\ref{sec:exp}). }

\emph{
	On the other hand, it is easy to verify from the factor-revealing LP for Theorem~\ref{THM:GREEDY} that if weight of $M_1$ is $o(\opt(G))$, then in fact the solution returned by the algorithm
	is a $\Paren{2+o(1)}$-approximation as opposed to only $3$ (simply change LHS in Constraint~(\ref{const:2}) to $o(1)$ instead). In particular, for ``more realistic'' values of $\ww{M_1}$, i.e., around $0.1$, the \emph{theoretical} guarantee of this 
	algorithm already gets within a factor 90\% of the sequential greedy algorithm (i.e., the approximation ratio becomes $2.22$).
}
\end{remark}

%% file: app-exp.tex
\section{Empirical Study}\label{sec:exp}

We report the results of evaluating our algorithm on
a number of publicly available datasets.  

\ssection{Datasets.} We use different datasets with varying
sizes, ranging from about six million to half a trillion
edges. Table~\ref{tab:dataset-size} provides the statistics. 
\begin{table*}[b!]
\caption{Statistics about datasets used for empirical evaluation.\label{tab:dataset-size}}
\vspace{1mm}
\centering
\begin{tabular}{lrrr}
\hline
{\bf Dataset} & {\bf Number of vertices} & {\bf Number of edges} &
                                                                   {\bf Maximum degree} \\
\hline
{Friendster} & 65,608,366	& 546,396,770,507	& 2,151,462 \\
{Orkut} & 3,072,441 &	21,343,527,822 &	893,056 \\
{LiveJournal} & 4,846,609 &	3,930,691,845 &	444,522 \\
{DBLP} & 1,482,070 &	5,946,953 &	1,961 \\
\hline
\end{tabular}
\end{table*}
We remark that the number of edges is half the sum of vertex degrees, whereas some previous work (e.g., \cite{BBDHKLM17}) report
the latter as the number of edges. Three of our datasets (namely Friendster, Orkut, and LiveJournal, all taken from SNAP) are the {\em public} datasets used in evaluating a
hierarchical clustering algorithm in~\cite{BBDHKLM17} (they also have a fourth {\em private} dataset). Maximum weight matching is important in the context of hierarchical
clustering, because it can be used to generate a more balanced
hierarchy:  Iteratively find a maximum weight matching and contract
the edges of the matching.  Then the size of clusters at each level $k$ will be
nearly the same and equal to $2^k$. We also add
a fourth dataset of our own based on \emph{public} DBLP co-authorship~\cite{DBLP:MIT}: vertices
denote authors and edge weights denote the number of papers two
researchers co-authored.

\ssection{Implementation Details.} %Scalability: We run on big data.
We implement our algorithm on a real proprietary MapReduce-like cluster, i.e., no simulations, with tens to hundreds of machines,
depending on the size of the dataset. In our set-up, as in any standard MapReduce system, data is
initially stored in distributed storage (each edge on an arbitrary machine), and is then fed into the system and the output matching will be stored in a single file on one machine. To be more precise, our implementation follows
the approach in Theorem~\ref{THM:GREEDY}: we first compute a coreset on each machine using $\GreedyCoreset$ and then after combining the edges, we simply run the greedy algorithm
again to compute the final solution. As such, our algorithm can be seen as the distributed version of the 
{\em sequential} greedy algorithm for the weighted matching problem. Our
experiments verify that our coreset approach provides significant speed-ups with little quality loss compared to  sequential greedy , which complements our theoretical results in Theorems~\ref{thm:main} and~\ref{THM:GREEDY}.

The two main measures we are interested in are the \emph{quality} of the returned solution and the \emph{speed-up} of the algorithms compared to the sequential greedy algorithm: 

\begin{itemize}
	\item \emph{Quality.} Our coreset-based approach typically secures close to 99.5\% of the weight of the solution found by the sequential greedy algorithm (and slightly above this for the cardinality instead of weight). 
Table~\ref{tab:speedup} shows the quality of the solution obtained on each dataset separately. We emphasize that even though our theoretical proof in Theorem~\ref{THM:GREEDY} guarantees
 a 66\% performance (for the
coreset-based algorithm with greedy as post-processing compared to the sequential greedy
algorithm), we never lose more than 1--2\% in terms of the solution quality (recall also Remark~\ref{rem:too-good} for a theoretical explanation of this phenomenon).

	\item \emph{Speed-up.} The speed-up achieved by our coreset-based approach varies significantly between the datasets---ranging from $3$x to $130$x---compared to the sequential greedy algorithm. The speed-up is larger for bigger graphs, where the coreset
	computation can really take advantage of the parallelism. Table~\ref{tab:speedup} provides the speed-up obtained on each dataset separately. 
\end{itemize}

 \begin{table}[h!]
\caption{Speed-up and relative solution quality of our coreset-based
  algorithm compared to the sequential greedy algorithm.\label{tab:speedup}}
\vspace{1mm}
\centering
\begin{tabular}{lccc}
\hline
{\bf Dataset} & {\bf Speed-up} & {\bf Weight} &
                                                                   {\bf Cardinality} \\
\hline
{Friendster} & 130x & 99.83\% & 99.74\% \\
{Orkut} & \phantom{0}17x & 98.90\% & 99.42\% \\
{LiveJournal} & \phantom{00}6x & 99.86\% & 99.79\% \\
{DBLP} & \phantom{00}3x & 99.55\% & 99.27\% \\
\hline
\end{tabular}
\end{table}

Several remarks are in order. Firstly, since the graphs in our datasets are too big, we could not compute the value of optimal weighted matching on these graphs.
Therefore we only compare the quality of our solution to the quality of the sequential greedy algorithm that does not use a coreset. Additionally, in Table~\ref{tab:speedup}, the actual running times and number of machines are withheld due 
to company's policy so as to not reveal specifics of hardware and architecture in our company's clusters. Instead we focus on speed-up (as did several previous work) which is, unlike absolute running
time, mostly independent of computing infrastructure, and thus a better performance measure for the algorithms. Indeed, we expect that any MapReduce-based system to see similar speed-up to what we report. 
However, we provide the following back-of-the-envelope calculation in order to help the reader estimate the general whereabouts of the running times.
In the biggest dataset we consider, namely Friendster, storing the graph in the standard
edge format used by the Stanford Network Analysis Project (SNAP)
requires 20--30TB.  Merely reading this data into memory from a
typical 7200 RPM HDD (with maximum read speed of 160MB/s) takes at
least 35--50 hours. On this dataset, we obtain about 130x speed-up by building a coreset
and running the greedy algorithm on it.  Our solution secures 99.83\%
of the weight of the solution found by the sequential greedy algorithm, and its cardinality
is 99.74\% that of the greedy's.

\ssection{Role of the Multiplicity Parameter.} 
%Increase multiplicity.
We also verify empirically that increasing the multiplicity parameter
in coreset computation improves the overall quality of the final
matching solution.  This holds for both the cardinality and the weight of
the matching obtained.  We demonstrate this for two values of $k$,
which is the number of subgraphs we partition the original graph into in the coreset approach (which is also the number of machines).
Our results in this part are summarized in Figure~\ref{fig:mult}. 

\begin{figure}[h!]
\centering
\includegraphics[width=9cm]{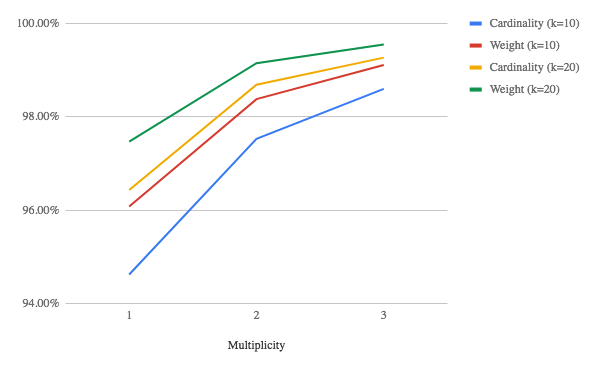}
\caption{The effect of increasing multiplicity on the quality of the solution.\label{fig:mult} The $y$-axis numbers are relative to
  the weight and cardinality of the sequential greedy solution and $k$ denotes the number of subgraphs we partition the original graph into in the coreset approach.}
\end{figure}

\subsection{Comparison with Prior Algorithms}\label{sec:comp}

We provide an empirical comparison of our algorithm with two prior algorithms, a simple and natural sampling based algorithm, and the coreset-based approach of Assadi and Khanna~\cite{AssadiK17} for unweighted matching 
and its extension to weighted matchings using the Crouch-Stubbs technique~\cite{CrouchS14}. At the end of this section, we also present some further theoretical comparison of our algorithm with  prior algorithms.  

\ssection{Empirical comparison with Sampling Algorithm.} 
We compare our coreset approach with a simple {sampling based} algorithm that first samples a fraction of edges of the input graph, 
and then computes the weighted matching only on these edges as opposed to the whole graph itself. This is in fact equivalent to comparing the weight of the matching returned as a coreset on one of the machines, 
as the subgraph on each machine is simply a random subgraph of the input graph and the coreset is just the weighted greedy matching.
As shown in Figure~\ref{fig:sampling}, the quality of the best greedy matching
(among all subgraphs in the random clustering) degrades as we increase 
the number of pieces.  In contrast, the quality of the final solution
improves as we increase the number of pieces as was expected by Remark~\ref{rem:too-good}.
Our experimental results in Figure~\ref{fig:sampling} demonstrate the power of our algorithm compared to this simple approach towards implementing
the sequential greedy algorithm on massive graphs. 

\begin{figure}[h!]
\centering
\includegraphics[width=9cm]{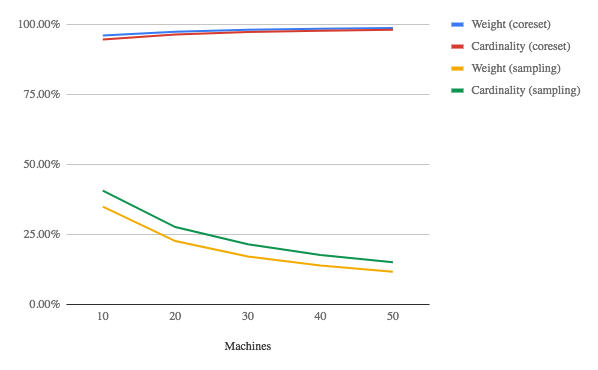}
\caption{Sampling-based algorithm vs.\ the coreset-based
  approach.\label{fig:sampling} The $y$-axis numbers are relative to
  the weight and cardinality of the sequential greedy solution.}
\end{figure}

\ssection{Empirical Comparison with Assadi-Khanna Algorithm~\cite{AssadiK17}.} To provide a more complete picture, we also implemented the previous coreset-based algorithm (for the unweighted matching) by Assadi
 and Khanna~\cite{AssadiK17} as well as its extension to the weighted setting via the Crouch-Stubbs technique~\cite{CrouchS14}.

Overall we learn that our new method provides more speed-up and better solution quality (i.e., the weight of the matching).  The weight of the matching produced via the previous method is a few percentage points worse than what our method produces.  
This is partly because the previous method worked on \emph{unweighted} graphs and we had to extend it to weighted matching algorithm via the Crouch-Stubbs technique that theoretically loses a factor 2 (see~\cite{CrouchS14}). 

The speed-up is an order of magnitude smaller for the two bigger graphs: 4.72 vs.\ 131.58 for Friendster and 1.36 vs.\ 17.13 for Orkut.  
For the two smaller graphs the speed-up of Assadi-Khanna algorithm is about a factor 2 less than our new method.
See Table~\ref{tab:speedup-AK} 
for exact numbers. 

\begin{table}[h!]
\caption{Speed-up and relative solution quality of the Assadi-Khanna algorithm combined with the Crouch-Stubbs technique compared to the sequential greedy algorithm. The results here should be contrasted with our bounds in Table~\ref{tab:speedup}. \label{tab:speedup-AK}}
\vspace{1mm}
\centering
\begin{tabular}{lccc}
\hline
{\bf Dataset} & {\bf Speed-up} & {\bf Weight} &
                                                                   {\bf Cardinality} \\
\hline
{Friendster} & 4.72x & 94.43\% & 98.34\% \\
{Orkut} & \phantom{0}1.36x & 92.41\% & 99.71\% \\
{LiveJournal} & \phantom{0}2.54x & 96.52\% & 99.36\% \\
{DBLP} & \phantom{0}1.19x & 92.42\% & 92.67\% \\
\hline
\end{tabular}
\end{table}

One drawback of the Assadi-Khanna algorithm is that it needs to find the \emph{maximum} matching in the first stage as opposed to a \emph{greedy} matching in our method (see~\cite{AssadiK17} for further discussion on necessity of this in their approach). 
 The maximum matching algorithms in non-bipartite graphs rely on the structure of the so-called blossoms, and (even in the case of unweighted graphs) are much more complicated than the greedy algorithm; they are also much slower---more than
 quadratic time vs.\ linear time.  Therefore, one needs to more carefully balance the number of pieces the graph is chopped into during the coreset stage (so as not to make it computationally prohibitive), but also to avoid having too many edges in the second stage 
 where we want to run a single-machine maximum matching algorithm.

For the LiveJournal graph, we ran more experiments to understand the effect of the parameter $\epsilon$ from Crouch-Stubbs reduction.  Recall that the overhead of their technique to turn an unweighted matching algorithm into a weighted matching algorithm is, 
theoretically, a factor $2+\epsilon$.  In our case, changing $\epsilon$ from $1.0$ down to $0.01$ did not affect the running time significantly.  However, it has a major impact on the solution quality.  The relative weight of the produced matching improved from 87.58\% to 
96.52\%.  We further note that not using the Crouch-Stubbs technique (and merely running the unweighted algorithm of Assadi and Khanna) resulted in a solution of relative weight only 22.16\% (which is to be expected as the original algorithm of~\cite{AssadiK17} is 
for the unweighted matching problem only and is  oblivious to the weight of the edges).

\ssection{Theoretical comparison with prior constant round MapReduce algorithms.} Prior MapReduce algorithms for weighted matching that use constant number of rounds (and $n^{1+\Omega(1)}$ space comparable to ours) 
include~\cite{LattanziMSV11,CrouchS14,AssadiK17,AssadiBBMS17} that also require a small number of rounds, typically between 2 to 10 rounds, and~\cite{AhnG15,HarveyLL18} that require a large unspecified constant number of rounds.  
The approximation ratio of our algorithm improves upon all algorithms in the first category while using the minimal number of 2 rounds but matches that of~\cite{HarveyLL18} and is worse than~\cite{AhnG15} in the second category. 
However, the algorithms in the second category are considerably more complicated than our algorithm (in particular the $(1+\eps)$-approximation algorithm of~\cite{AhnG15}) and also require a very large number of rounds of computation in the MapReduce model
and are hence less amenable to efficient implementation in practice. 

\ssection{Theoretical comparison with prior MapReduce algorithms with smaller memory.} There are also MapReduce algorithms known in the literature for weighted matching that require a smaller memory per machine than our algorithm (typically just near-linear in $n$ memory) at a cost of super-constant (typically $O(\log\log{n})$) number of rounds~\cite{CzumajLMMOS18,AssadiBBMS17,GhaffariGKMR18,GamlathKMS18}. 
However, considering the fact that these algorithms are all quite complicated on one hand and also have an exponential dependence on $\eps$ in both number of rounds and memory of algorithms on the other hand\footnote{For instance, just to obtain an approximation ratio of $2.1$ for weighted matching (i.e., for $\eps = 0.1$), a basic calculation of the bounds in these algorithms implies that the number of required rounds is $> 10^{10}$ rounds (even entirely ignoring the $\log\log{n}$ and other hidden constant terms).} (as opposed to near-linear dependence and only for memory in our algorithm) suggests that barring a highly optimized and careful implementation of these algorithms (which is well beyond the scope of our paper), such algorithms may not be amenable to efficient implementation in practice.

%% file: app-optimal.tex
\section{Optimality of Our Coreset Size}\label{app:coreset-lower}
\newcommand{\FG}{\ensuremath{\mathcal{G}}}

In this section, we prove the following theorem that states that size of our coreset in Theorem~\ref{thm:weighted-matching-0.5} is optimal among \emph{any constant factor} randomized coreset with \emph{constant} expected multiplicity. 
The proof is based an adaption of the argument in~\cite{AssadiK17} which only works for coresets with \emph{exact} multiplicity $1$. 
\begin{theorem}\label{thm:coreset-optimal}
	Let $\alg$ be any \emph{constant} factor randomized composable coreset for maximum cardinality matching (and hence for maximum weight matching also) with constant expected multiplicity $O(1)$. Then size of the coreset computed by $\alg$ on each subgraph is $\Omega(n)$. 
\end{theorem}
\begin{proof}
	Let $\alpha = O(1)$ denotes the approximation ratio of $\alg$ and $c = O(1)$ be the expected multiplicity. Also, let $k = \omega(\alpha \cdot c)$ be the number of subgraphs in the random clustering (note that $k$ needs to be larger than $\alpha \cdot c$ as 
	obtaining $(k/c)$-approximation is trivial by returning maximum weight matching on each subgraph as the coreset).  
	
	Consider the following family $\FG$ of bipartite graphs $G(L,R,E)$ defined as follows: 
	\begin{enumerate}[leftmargin=20pt]
		\item Sample two sets $A \subseteq L$ and $B \subseteq R$, each of size $n_{AB} := \gamma \cdot \frac{n}{c \cdot \alpha}$ vertices (for some constant $\gamma \in [0,1]$ to be determined later). 
		\item Sample each edge in $A \times B$ with probability $p_{AB} := \frac{k}{c \cdot n_{AB}}$. Let $E_{AB}$ be these edges. 
		\item Sample a uniformly at random chosen \emph{perfect} matching between $\overline{A}$ and $\overline{B}$. Let $M_{\overline{AB}}$ be this matching. 
		\item Graph $G(L,R,E)$ is defined by setting $E:= E_{AB} \cup M_{\overline{AB}}$. 
	\end{enumerate}
	
	For $G \sim \FG$, we further define the subgraph $G_{AB}(A,B,E_{AB})$ as the \emph{induced} subgraph of $G$ on vertices $A$ and $B$. 
	
	Now consider the random $k$-partitioning with expected multiplicity $c$ for a graph $G \sim \FG$, resulting in $\Gi{1},\ldots,\Gi{k}$.  
	We further define $\Gi{i}_{AB}$ as the induced subgraph of $\Gi{i}$ on vertices $A$ and $B$, i.e., the sampled part of $G_{AB}$ in $\Gi{i}$. Finally, define $M_{AB}^{(i)}$ as the \emph{largest induced matching} in $\Gi{i}_{AB}$. We have the following 
	rather standard lemma similar to~\cite{AssadiK17}. 
	
	\begin{lemma}\label{lem:standard-induced}
		With high probability, for all $i \in [k]$, $\card{M^{(i)}} = \Theta(n_{AB})$. 
	\end{lemma}
	\begin{proof}
		By definition, each $\Gi{i}_{AB}$ is obtained from $G_{AB}$ by sampling each edge from $G$ with probability $c/k$. By considering both the distribution of $\Gi{i}_{AB}$ with respect to $G_{AB}$ and the distribution of $G_{AB}$ itself, 
		we obtain that $\Gi{i}_{AB}$ is a \emph{random} bipartite graph with $n_{AB}$ vertices on each side with probability of each appearing equal to $1/n_{AB}$. The lemma then follows from standard arguments in random
		graph theory; see, e.g. the text by Bollobas~\cite{Bollobas98} (see also~\cite[Lemma A.3]{AssadiK17} for a direct proof of this statement). 
	\end{proof}
	
	We are now ready to prove Theorem~\ref{thm:coreset-optimal}. Define $M^{(i)}$ as the \emph{largest induced matching} in $\Gi{i}$ (not only $\Gi{i}_{AB}$ as in definition of $M^{(i)}_{AB}$). It is immediate to verify that 
	$M^{(i)} = M^{(i)}_{AB} \cup M^{(i)}_{\overline{AB}}$, where $M^{(i)}_{\overline{AB}}$ is the sampled part of $M_{\overline{AB}}$ in $\Gi{i}$, which by Chernoff bound, has size $M^{(i)}_{\overline{AB}} = \Theta(n \cdot c/k)$ with high probability. In the 
	following, we condition on this event and the events in Lemma~\ref{lem:standard-induced} for all the $k$ subgraphs, which still happens with high probability by union bound (and hence does not change the expected performance of the algorithm 
	by more than
	an \emph{additive} $o(1)$ term which we are going to ignore for simplicity). 
	
	The important observation is that $\alg$ is ``oblivious'' to the distinction between
	$M^{(i)}_{AB}$ and $M^{(i)}_{\overline{AB}}$ in $M^{(i)}$. More formally, conditioned on a choice of $M^{(i)}$ and the remainder of graph $\Gi{i}$, the partitioning of $M^{(i)}$ to $M^{(i)}_{AB} \cup M^{(i)}_{\overline{AB}}$ is  
	such that each edge appears in $M^{(i)}_{\overline{AB}}$ with probability: 
	\begin{align*}
		p := \frac{\card{M^{(i)}_{\overline{AB}}}}{{\card{M^{(i)}}}} = \Theta(\frac{n \cdot c^2 \cdot \alpha}{k \cdot \gamma \cdot n}) = \Theta(\frac{c^2 \cdot \alpha}{\gamma \cdot k}). 
	\end{align*}
	
	However, conditioned on a choice of $\Gi{i}$, the output of $\alg$ is fixed as it is only a function of the input graph. Let $C_i$ be the coreset returned by $\alg$ on $\Gi{i}$. We have, 
	\begin{align*}
		\Ex{\card{C_i \cap M^{(i)}_{\overline{AB}}}} &= \sum_{e \in C_i} \Pr\paren{e \in M^{(i)}_{\overline{AB}} \mid e \in M^{(i)}} = \card{C_i} \cdot \Theta(\frac{c^2 \cdot \alpha}{\gamma \cdot k}). 
	\end{align*}
	
	Adding up this quantity across all coresets $C := C_1 \cup \ldots \cup C_k$, we have,
	\begin{align*}
		\Ex{\card{C \cap M_{\overline{AB}}}} \leq \sum_{i=1}^{k} \card{C_i} \cdot \Theta(\frac{c^2 \cdot \alpha}{\gamma \cdot k}) = s \cdot \Theta(\frac{c^2 \cdot \alpha}{\gamma}), 	
	\end{align*}
	where $s$ denotes the size of the coreset. Now suppose towards a contradiction that $s \leq \frac{n}{\gamma^2 \cdot c^2 \cdot \alpha^2}$. By above calculation, we have,
	\begin{align*}
		\Ex{\card{C \cap M_{\overline{AB}}}} \leq \Theta(\gamma \cdot n/\alpha) \leq n/4\alpha,
	\end{align*}
	by taking $\gamma$ to be a sufficiently small constant. 
	
	We are nearly done. Graph $G$ always have a matching of size at least $n-n_{AB} \geq n/2$, by taking  $M_{AB}$. 
	On the other hand, the expected intersection of $\card{C \cap M_{\overline{AB}}} \leq n/4\alpha$ by above calculation. Moreover, the remaining edges in graph $G \setminus M_{\overline{AB}}$ are all incident
	on $2  n_{AB}$ vertices and hence can only contribute to the maximum matching in $G$ by $2\gamma \cdot\frac{n}{c \cdot \alpha} < n/4\alpha$ for $\gamma < 1/8$. This means that the expected size of the 
	maximum matching in union of all coresets, i.e., $C$ has size smaller than $n/2\alpha$, which is contradiction with the fact that the coreset contains an $\alpha$-approximate matching. 
	To conclude, we have that $s \geq \frac{n}{\gamma^2 \cdot c^2 \cdot \alpha^2} = \Omega(\frac{n}{c^2 \cdot \alpha^2})$ which is $\Omega(n)$ for $c,\alpha = \Theta(1)$. 
\end{proof}

%% file: main.bbl
\begin{thebibliography}{10}

\bibitem{AhnG13}
K.~J. Ahn and S.~Guha.
\newblock Linear programming in the semi-streaming model with application to
  the maximum matching problem.
\newblock {\em Inf. Comput.}, 222:59--79, 2013.

\bibitem{AhnG15}
K.~J. Ahn and S.~Guha.
\newblock Access to data and number of iterations: Dual primal algorithms for
  maximum matching under resource constraints.
\newblock In {\em Proceedings of the 27th {ACM} on Symposium on Parallelism in
  Algorithms and Architectures, {SPAA} 2015, Portland, OR, USA, June 13-15,
  2015}, pages 202--211, 2015.

\bibitem{AhnGM12}
K.~J. Ahn, S.~Guha, and A.~McGregor.
\newblock Analyzing graph structure via linear measurements.
\newblock In {\em Proceedings of the Twenty-third Annual ACM-SIAM Symposium on
  Discrete Algorithms}, SODA '12, pages 459--467, 2012.

\bibitem{AssadiBBMS17}
S.~Assadi, M.~Bateni, A.~Bernstein, V.~S. Mirrokni, and C.~Stein.
\newblock Coresets meet {EDCS:} algorithms for matching and vertex cover on
  massive graphs.
\newblock In {\em Proceedings of the Thirtieth Annual {ACM-SIAM} Symposium on
  Discrete Algorithms, {SODA} 2019, San Diego, California, USA, January 6-9,
  2019}, pages 1616--1635, 2019.

\bibitem{AssadiK17}
S.~Assadi and S.~Khanna.
\newblock Randomized composable coresets for matching and vertex cover.
\newblock In {\em Proceedings of the 29th {ACM} Symposium on Parallelism in
  Algorithms and Architectures, {SPAA} 2017, Washington DC, USA, July 24-26,
  2017}, pages 3--12, 2017.

\bibitem{AssadiKL17}
S.~Assadi, S.~Khanna, and Y.~Li.
\newblock On estimating maximum matching size in graph streams.
\newblock In {\em Proceedings of the Twenty-Eighth Annual {ACM-SIAM} Symposium
  on Discrete Algorithms, {SODA} 2017, Barcelona, Spain, Hotel Porta Fira,
  January 16-19}, pages 1723--1742, 2017.

\bibitem{AssadiKLY16}
S.~Assadi, S.~Khanna, Y.~Li, and G.~Yaroslavtsev.
\newblock Maximum matchings in dynamic graph streams and the simultaneous
  communication model.
\newblock In {\em Proceedings of the Twenty-Seventh Annual {ACM-SIAM} Symposium
  on Discrete Algorithms, {SODA} 2016, Arlington, VA, USA, January 10-12,
  2016}, pages 1345--1364, 2016.

\bibitem{BBDHKLM17}
M.~Bateni, S.~Behnezhad, M.~Derakhshan, M.~Hajiaghayi, R.~Kiveris, S.~Lattanzi,
  and V.~S. Mirrokni.
\newblock Affinity clustering: Hierarchical clustering at scale.
\newblock In {\em 30th Annual Conference on Neural Information Processing
  Systems}, pages 6867--6877, 2017.

\bibitem{BeameKS13}
P.~Beame, P.~Koutris, and D.~Suciu.
\newblock Communication steps for parallel query processing.
\newblock In {\em Proceedings of the 32nd {ACM} {SIGMOD-SIGACT-SIGART}
  Symposium on Principles of Database Systems, {PODS} 2013, New York, NY, {USA}
  - June 22 - 27, 2013}, pages 273--284, 2013.

\bibitem{BehnezhadDETY17}
S.~Behnezhad, M.~Derakhshan, H.~Esfandiari, E.~Tan, and H.~Yami.
\newblock Brief announcement: Graph matching in massive datasets.
\newblock In {\em Proceedings of the 29th {ACM} Symposium on Parallelism in
  Algorithms and Architectures, {SPAA} 2017, Washington DC, USA, July 24-26,
  2017}, pages 133--136, 2017.

\bibitem{BehnezhadR18}
S.~Behnezhad and N.~Reyhani.
\newblock Almost optimal stochastic weighted matching with few queries.
\newblock {\em CoRR}, abs/1710.10592. To appear in EC 2018, 2017.

\bibitem{BSX08}
B.~Berger, R.~Singh, and J.~Xu.
\newblock Graph algorithms for biological systems analysis.
\newblock In {\em Proceedings of the Nineteenth Annual {ACM-SIAM} Symposium on
  Discrete Algorithms, {SODA} 2008}, pages 142--151, 2008.

\bibitem{BlumDHPSS15}
A.~Blum, J.~P. Dickerson, N.~Haghtalab, A.~D. Procaccia, T.~Sandholm, and
  A.~Sharma.
\newblock Ignorance is almost bliss: Near-optimal stochastic matching with few
  queries.
\newblock In {\em Proceedings of the Sixteenth {ACM} Conference on Economics
  and Computation, {EC} '15, Portland, OR, USA, June 15-19, 2015}, pages
  325--342, 2015.

\bibitem{Bollobas98}
B.~Bollob{\'a}s.
\newblock Random graphs.
\newblock In {\em Modern graph theory}, pages 215--252. Springer, 1998.

\bibitem{CharlesCDJS10}
D.~X. Charles, M.~Chickering, N.~R. Devanur, K.~Jain, and M.~Sanghi.
\newblock Fast algorithms for finding matchings in lopsided bipartite graphs
  with applications to display ads.
\newblock In {\em Proceedings 11th {ACM} Conference on Electronic Commerce
  (EC-2010), Cambridge, Massachusetts, USA, June 7-11, 2010}, pages 121--128,
  2010.

\bibitem{CrouchS14}
M.~Crouch and D.~S. Stubbs.
\newblock Improved streaming algorithms for weighted matching, via unweighted
  matching.
\newblock In {\em Approximation, Randomization, and Combinatorial Optimization.
  Algorithms and Techniques, {APPROX/RANDOM} 2014, September 4-6, 2014,
  Barcelona, Spain}, pages 96--104, 2014.

\bibitem{CyganGS12}
M.~Cygan, H.~N. Gabow, and P.~Sankowski.
\newblock Algorithmic applications of baur-strassen's theorem: Shortest cycles,
  diameter and matchings.
\newblock In {\em 53rd Annual {IEEE} Symposium on Foundations of Computer
  Science, {FOCS} 2012, New Brunswick, NJ, USA, October 20-23, 2012}, pages
  531--540, 2012.

\bibitem{CzumajLMMOS18}
A.~Czumaj, J.~Lacki, A.~Madry, S.~Mitrovic, K.~Onak, and P.~Sankowski.
\newblock Round compression for parallel matching algorithms.
\newblock In {\em Proceedings of the 50th Annual {ACM} {SIGACT} Symposium on
  Theory of Computing, {STOC} 2018, Los Angeles, CA, USA, June 25-29, 2018},
  pages 471--484, 2018.

\bibitem{BarbosaENW15}
R.~da~Ponte~Barbosa, A.~Ene, H.~L. Nguyen, and J.~Ward.
\newblock The power of randomization: Distributed submodular maximization on
  massive datasets.
\newblock In {\em Proceedings of the 32nd International Conference on Machine
  Learning, {ICML} 2015, Lille, France, 6-11 July 2015}, pages 1236--1244,
  2015.

\bibitem{BarbosaENW16}
R.~da~Ponte~Barbosa, A.~Ene, H.~L. Nguyen, and J.~Ward.
\newblock A new framework for distributed submodular maximization.
\newblock In {\em {IEEE} 57th Annual Symposium on Foundations of Computer
  Science, {FOCS} 2016, New Brunswick, New Jersey, {USA}}, pages 645--654,
  2016.

\bibitem{DBLP:MIT}
E.~Demaine and M.~Hajiaghayi, 2014.

\bibitem{DickersonPS12}
J.~P. Dickerson, A.~D. Procaccia, and T.~Sandholm.
\newblock Optimizing kidney exchange with transplant chains: theory and
  reality.
\newblock In {\em International Conference on Autonomous Agents and Multiagent
  Systems, {AAMAS} 2012, Valencia, Spain, June 4-8, 2012 {(3} Volumes)}, pages
  711--718, 2012.

\bibitem{DuanP10}
R.~Duan and S.~Pettie.
\newblock Approximating maximum weight matching in near-linear time.
\newblock In {\em 51th Annual {IEEE} Symposium on Foundations of Computer
  Science, {FOCS} 2010, October 23-26, 2010, Las Vegas, Nevada, {USA}}, pages
  673--682, 2010.

\bibitem{DuanP14}
R.~Duan and S.~Pettie.
\newblock Linear-time approximation for maximum weight matching.
\newblock {\em J. {ACM}}, 61(1):1:1--1:23, 2014.

\bibitem{DuanPS17}
R.~Duan, S.~Pettie, and H.~Su.
\newblock Scaling algorithms for weighted matching in general graphs.
\newblock In {\em Proceedings of the Twenty-Eighth Annual {ACM-SIAM} Symposium
  on Discrete Algorithms, {SODA} 2017, Barcelona, Spain, Hotel Porta Fira,
  January 16-19}, pages 781--800, 2017.

\bibitem{DK99}
I.~S. Duff and J.~Koster.
\newblock The design and use of algorithms for permuting large entries to the
  diagonal of sparse matrices.
\newblock {\em {SIAM} J. Matrix Analysis Applications}, 20(4):889--901, 1999.

\bibitem{DK01}
I.~S. Duff and J.~Koster.
\newblock On algorithms for permuting large entries to the diagonal of a sparse
  matrix.
\newblock {\em {SIAM} J. Matrix Analysis Applications}, 22(4):973--996, 2001.

\bibitem{Edmonds65a}
J.~Edmonds.
\newblock Maximum matching and a polyhedron with 0, 1-vertices.
\newblock {\em Journal of Research of the National Bureau of Standards B},
  69(125-130):55--56, 1965.

\bibitem{Edmonds65}
J.~Edmonds.
\newblock Paths, trees, and flowers.
\newblock {\em Canadian Journal of mathematics}, 17(3):449--467, 1965.

\bibitem{EpsteinLMS11}
L.~Epstein, A.~Levin, J.~Mestre, and D.~Segev.
\newblock Improved approximation guarantees for weighted matching in the
  semi-streaming model.
\newblock {\em {SIAM} J. Discrete Math.}, 25(3):1251--1265, 2011.

\bibitem{Gabow76}
H.~N. Gabow.
\newblock An efficient implementation of edmonds' algorithm for maximum
  matching on graphs.
\newblock {\em J. {ACM}}, 23(2):221--234, 1976.

\bibitem{Gabow85}
H.~N. Gabow.
\newblock A scaling algorithm for weighted matching on general graphs.
\newblock In {\em 26th Annual Symposium on Foundations of Computer Science,
  Portland, Oregon, USA, 21-23 October 1985}, pages 90--100, 1985.

\bibitem{Gabow90}
H.~N. Gabow.
\newblock Data structures for weighted matching and nearest common ancestors
  with linking.
\newblock In {\em Proceedings of the First Annual {ACM-SIAM} Symposium on
  Discrete Algorithms, 22-24 January 1990, San Francisco, California.}, pages
  434--443, 1990.

\bibitem{GabowGS84}
H.~N. Gabow, Z.~Galil, and T.~H. Spencer.
\newblock Efficient implementation of graph algorithms using contraction.
\newblock In {\em 25th Annual Symposium on Foundations of Computer Science,
  West Palm Beach, Florida, USA, 24-26 October 1984}, pages 347--357, 1984.

\bibitem{GabowT91}
H.~N. Gabow and R.~E. Tarjan.
\newblock Faster scaling algorithms for general graph-matching problems.
\newblock {\em J. {ACM}}, 38(4):815--853, 1991.

\bibitem{GamlathKMS18}
B.~Gamlath, S.~Kale, S.~Mitrovic, and O.~Svensson.
\newblock Weighted matchings via unweighted augmentations.
\newblock {\em CoRR}, abs/1811.02760. To appear in PODC 2019., 2018.

\bibitem{GamlathKMSW19}
B.~Gamlath, M.~Kapralov, A.~Maggiori, O.~Svensson, and D.~Wajc.
\newblock Online matching with general arrivals.
\newblock {\em CoRR}, abs/1904.08255, 2019.

\bibitem{GhaffariGKMR18}
M.~Ghaffari, T.~Gouleakis, C.~Konrad, S.~Mitrovic, and R.~Rubinfeld.
\newblock Improved massively parallel computation algorithms for mis, matching,
  and vertex cover.
\newblock In {\em Proceedings of the 2018 {ACM} Symposium on Principles of
  Distributed Computing, {PODC} 2018, Egham, United Kingdom, July 23-27, 2018},
  pages 129--138, 2018.

\bibitem{GoelKK12}
A.~Goel, M.~Kapralov, and S.~Khanna.
\newblock On the communication and streaming complexity of maximum bipartite
  matching.
\newblock In {\em Proceedings of the Twenty-third Annual ACM-SIAM Symposium on
  Discrete Algorithms}, SODA '12, pages 468--485, 2012.

\bibitem{GoodrichSZ11}
M.~T. Goodrich, N.~Sitchinava, and Q.~Zhang.
\newblock Sorting, searching, and simulation in the mapreduce framework.
\newblock In {\em Algorithms and Computation - 22nd International Symposium,
  {ISAAC} 2011, Yokohama, Japan, December 5-8, 2011. Proceedings}, pages
  374--383, 2011.

\bibitem{gurobi}
Gurobi optimizer.
\newblock \url{http://www.gurobi.com}.
\newblock Accessed: May 2018.

\bibitem{HarveyLL18}
N.~J.~A. Harvey, C.~Liaw, and P.~Liu.
\newblock Greedy and local ratio algorithms in the mapreduce model.
\newblock In {\em Proceedings of the 30th on Symposium on Parallelism in
  Algorithms and Architectures, {SPAA} 2018, Vienna, Austria, July 16-18,
  2018}, pages 43--52, 2018.

\bibitem{HuangRVZ15}
Z.~Huang, B.~Radunovic, M.~Vojnovic, and Q.~Zhang.
\newblock Communication complexity of approximate matching in distributed
  graphs.
\newblock In {\em 32nd International Symposium on Theoretical Aspects of
  Computer Science, {STACS} 2015, March 4-7, 2015, Garching, Germany}, pages
  460--473, 2015.

\bibitem{JebaraWC09}
T.~Jebara, J.~Wang, and S.~Chang.
\newblock Graph construction and \emph{b}-matching for semi-supervised
  learning.
\newblock In {\em Proceedings of the 26th Annual International Conference on
  Machine Learning, {ICML} 2009, Montreal, Quebec, Canada, June 14-18, 2009},
  pages 441--448, 2009.

\bibitem{Kapralov13}
M.~Kapralov.
\newblock Better bounds for matchings in the streaming model.
\newblock In {\em Proceedings of the Twenty-Fourth Annual {ACM-SIAM} Symposium
  on Discrete Algorithms, {SODA} 2013, New Orleans, Louisiana, USA, January
  6-8, 2013}, pages 1679--1697, 2013.

\bibitem{KarloffSV10}
H.~J. Karloff, S.~Suri, and S.~Vassilvitskii.
\newblock A model of computation for mapreduce.
\newblock In {\em Proceedings of the Twenty-First Annual {ACM-SIAM} Symposium
  on Discrete Algorithms, {SODA} 2010, Austin, Texas, USA, January 17-19,
  2010}, pages 938--948, 2010.

\bibitem{KK98}
G.~Karypis and V.~Kumar.
\newblock A fast and high quality multilevel scheme for partitioning irregular
  graphs.
\newblock {\em {SIAM} J. Scientific Computing}, 20(1):359--392, 1998.

\bibitem{KonradMM12}
C.~Konrad, F.~Magniez, and C.~Mathieu.
\newblock Maximum matching in semi-streaming with few passes.
\newblock In {\em Approximation, Randomization, and Combinatorial Optimization.
  Algorithms and Techniques - 15th International Workshop, {APPROX} 2012, and
  16th International Workshop, {RANDOM} 2012, Cambridge, MA, USA, August 15-17,
  2012. Proceedings}, pages 231--242, 2012.

\bibitem{KumarMVV13}
R.~Kumar, B.~Moseley, S.~Vassilvitskii, and A.~Vattani.
\newblock Fast greedy algorithms in mapreduce and streaming.
\newblock In {\em 25th {ACM} Symposium on Parallelism in Algorithms and
  Architectures, {SPAA} '13, Montreal, QC, Canada - July 23 - 25, 2013}, pages
  1--10, 2013.

\bibitem{LD04}
C.~J. Langmead and B.~R. Donald.
\newblock High-throughput {3D} structural homology detection via {NMR}
  resonance assignment.
\newblock In {\em 3rd International {IEEE} Computer Society Computational
  Systems Bioinformatics Conference, {CSB} 2004}, pages 278--289, 2004.

\bibitem{LattanziMSV11}
S.~Lattanzi, B.~Moseley, S.~Suri, and S.~Vassilvitskii.
\newblock Filtering: a method for solving graph problems in mapreduce.
\newblock In {\em {SPAA} 2011: Proceedings of the 23rd Annual {ACM} Symposium
  on Parallelism in Algorithms and Architectures, San Jose, CA, USA, June 4-6,
  2011 (Co-located with {FCRC} 2011)}, pages 85--94, 2011.

\bibitem{LiuV19}
P.~Liu and J.~Vondr{\'{a}}k.
\newblock Submodular optimization in the mapreduce model.
\newblock In {\em 2nd Symposium on Simplicity in Algorithms, SOSA@SODA 2019,
  January 8-9, 2019 - San Diego, CA, {USA}}, pages 18:1--18:10, 2019.

\bibitem{ManshadiAGKMS13}
F.~M. Manshadi, B.~Awerbuch, R.~Gemulla, R.~Khandekar, J.~Mestre, and M.~Sozio.
\newblock A distributed algorithm for large-scale generalized matching.
\newblock {\em {PVLDB}}, 6(9):613--624, 2013.

\bibitem{McGregor05}
A.~McGregor.
\newblock Finding graph matchings in data streams.
\newblock In {\em Approximation, Randomization and Combinatorial Optimization,
  Algorithms and Techniques, 8th International Workshop on Approximation
  Algorithms for Combinatorial Optimization Problems, {APPROX} 2005 and 9th
  InternationalWorkshop on Randomization and Computation, {RANDOM} 2005,
  Berkeley, CA, USA, August 22-24, 2005, Proceedings}, pages 170--181, 2005.

\bibitem{MSVV07}
A.~Mehta, A.~Saberi, U.~V. Vazirani, and V.~V. Vazirani.
\newblock Adwords and generalized online matching.
\newblock {\em J. {ACM}}, 54(5):22, 2007.

\bibitem{MirrokniZ15}
V.~S. Mirrokni and M.~Zadimoghaddam.
\newblock Randomized composable core-sets for distributed submodular
  maximization.
\newblock In {\em Proceedings of the Forty-Seventh Annual {ACM} on Symposium on
  Theory of Computing, {STOC} 2015, Portland, OR, USA, June 14-17, 2015}, pages
  153--162, 2015.

\bibitem{PazS17}
A.~Paz and G.~Schwartzman.
\newblock A (2 + {\eps})-approximation for maximum weight matching in the
  semi-streaming model.
\newblock In {\em Proceedings of the Twenty-Eighth Annual {ACM-SIAM} Symposium
  on Discrete Algorithms, {SODA} 2017, Barcelona, Spain, Hotel Porta Fira,
  January 16-19}, pages 2153--2161, 2017.

\bibitem{PennT00}
M.~Penn and M.~Tennenholtz.
\newblock Constrained multi-object auctions and b-matching.
\newblock {\em Inf. Process. Lett.}, 75(1-2):29--34, 2000.

\bibitem{PettieS04}
S.~Pettie and P.~Sanders.
\newblock A simpler linear time 2/3-epsilon approximation for maximum weight
  matching.
\newblock {\em Inf. Process. Lett.}, 91(6):271--276, 2004.

\bibitem{PCP06}
A.~Pinar, E.~Chow, and A.~Pothen.
\newblock Combinatorial algorithms for computing column space bases that have
  sparse inverses.
\newblock {\em Electronic Transactions on Numerical Analysis}, 22:122--145,
  2005.

\bibitem{PF90}
A.~Pothen and C.~Fan.
\newblock Computing the block triangular form of a sparse matrix.
\newblock {\em {ACM} Trans. Math. Softw.}, 16(4):303--324, 1990.

\bibitem{Preis99}
R.~Preis.
\newblock Linear time 1/2-approximation algorithm for maximum weighted matching
  in general graphs.
\newblock In {\em {STACS} 99, 16th Annual Symposium on Theoretical Aspects of
  Computer Science, Trier, Germany, March 4-6, 1999, Proceedings}, pages
  259--269, 1999.

\bibitem{VinkemeierH05}
D.~E.~D. Vinkemeier and S.~Hougardy.
\newblock A linear-time approximation algorithm for weighted matchings in
  graphs.
\newblock {\em {ACM} Trans. Algorithms}, 1(1):107--122, 2005.

\end{thebibliography}
